\documentclass[11pt,reqno]{amsart}
\usepackage{amsmath,amssymb,mathrsfs}
\usepackage{graphicx,cite}
\usepackage{subfig}
\usepackage{epstopdf}
\usepackage{graphics} 
\usepackage{epsfig} 
\usepackage{tikz}
\usepackage{paralist}
\usepackage{booktabs}
\usepackage{enumerate}
\usepackage{enumitem}
\usepackage{mdwlist}
\usepackage{algorithm,algorithmicx}

\newtheorem{theorem}{Theorem}[section]

\newtheorem{remark}[theorem]{Remark}

\numberwithin{equation}{section}

\newcommand{\mi}{\mathrm{i}}
\newcommand{\ms}{\mathbb{S}^2}
\newcommand{\fp}{\mathfrak{p}}
\newcommand{\fs}{\mathfrak{s}}

\newcommand{\kp}{\kappa_\mathfrak{p}}
\newcommand{\ks}{\kappa_\mathfrak{s}}

\newcommand{\bu}{\boldsymbol{u}}

\newcommand{\bv}{\boldsymbol{v}}

\newcommand{\bx}{\boldsymbol{x}}

\newcommand{\bd}{\boldsymbol{d}}
\newcommand{\bp}{\boldsymbol{p}}
\newcommand{\bJ}{\boldsymbol{J}} 
 
\newcommand{\uinc}{\boldsymbol{u^i}}
\newcommand{\uincf}{\boldsymbol{u}^{\boldsymbol{i}}_f}

\begin{document}

\title[Inverse elastic scattering of multiple particles in 3D]{Fast inverse elastic scattering of multiple particles in three dimensions}

\author{Jun Lai}
\address{School of Mathematical Sciences, Zhejiang University,
	Hangzhou, Zhejiang 310027, China}
\email{laijun6@zju.edu.cn}
\author{Jinrui Zhang}
\address{School of Mathematical Sciences, Zhejiang University,
	Hangzhou, Zhejiang 310027, China}
\email{12035013@zju.edu.cn}

\subjclass[2020]{35B40, 35P25, 65R20, 78A46}

\keywords{Elastic scattering, multiple scattering, inverse obstacle scattering, time reversal method, fast multipole method}

\begin{abstract}
Many applications require recovering the geometry information of multiple elastic particles based on the scattering information. In this paper, we consider the inverse time-harmonic elastic scattering of multiple rigid particles in three dimensions. We measure the far field information and apply the time reversal method to recover the unknown elastic particles. Two regimes are considered depending on the size and distance among particles. First, an asymptotic analysis for the imaging of small and distant particles is given based on the scattering property of a single particle, which can be used for selective focusing. Second, when particles are not small but well-separated, a fast algorithm, based on the combination of multiple scattering theory and fast multipole method, is proposed to efficiently simulate the forward multiple scattering problem and applied in the inverse elastic scattering. Numerical experiments  demonstrate the proposed method can determine the locations and shapes of multiple particles instantly.
\end{abstract}

\maketitle

\section{Introduction}
The inverse scattering problem for elastic waves has attracted extensive attentions due to its applications in nondestructive testing, medical imaging and seismic exploration \cite{LL-86, ABG-15, DLL2019, DLL2020}. This work is concerned with the time-harmonic inverse elastic scattering by multiple rigid particles embedded in a homogeneous and isotropic elastic medium in three dimensions. More specifically, given the elastic far field pattern from different incident directions, the goal is to recover the geometry information, including locations and shapes, of multiple unknown elastic scatterers. 

Over the years, many methods have been proposed to solve the inverse elastic obstacle problem, including both the iterative type and direct imaging methods. For instance, in \cite{L-IP15, L-SIAP12}, domain derivatives for the elastic scattering were derived based on boundary integral equations and variational techniques, which can be used to design gradient descent methods for reconstructing unknown elastic obstacles.  In \cite{AK2002}, the linear sampling method based on the factorization of far field operator has been studied for the inverse elastic obstacle  scattering.  Related work on the sampling methods for elastic scattering can also be found in \cite{TA01}. For the reconstruction of finitely many point-like and extended elastic obstacles,  factorization methods have been extensively studied in \cite{HKS-IP13,GAMT}. 

In this paper, we develop the time reversal method (TRM) to reconstruct the multiple elastic particles in three dimensions. The basic idea of TRM is to emit an incident wave into the non-dissipative elastic medium containing the unknown particles and then measure the corresponding far field. The measured field is then conjugated and re-emitted. From the far field data, it shows the unknown scatterers can be recovered by using the eigenvector of the far field operator as the kernel of a Herglotz incident wave. A special but important case is when particles are small and distant, the wave interaction among different particles can be ignored. Therefore, the far field is approximately given as a superposition of scattering from each individual particle, so the operator can be diagonalized by the eigenfunctions of each particle.  Based on the asymptotic analysis, we give the approximate eigenfunctions corresponding to significant eigenvalues for each particle,  which can be used to selectively focus a given particle. It is worth mentioning that the mathematical justification for TRM in the inverse acoustic and electromagnetic obstacle scattering has been given in \cite{CK04,APRT08}. Our result can be taken as an extension from the acoustic and electromagnetic scattering to the elastic scattering.

When the size of each particle is comparable to the wavelength of the incoming field and the distribution of particles is reasonably dense, the interaction of the particles involves non-negligible multiple scattering effects. To apply the time reversal method, a numerical algorithm for the forward problem is needed to obtain the scattered information from multiple particles. Practical algorithms  often rely on the addition theorem\cite{M-06} to transform the field from one particle to another. However, straightforward computation for the coefficients in the addition theorem is very expensive and suffers from numerical instability\cite{Xu1996}. In this situation, the application of TRM becomes very challenging as it requires the far field data from all incident directions and measured in all directions as well. Since the forward problem has to be solved many times, it is prohibitively difficult when the number of particles is large, especially in three dimensions. 

On the other hand, wave scattering from multiple particles is also referred as ``multiple scattering" in the literature, and has a very rich history in the  acoustics and electromagnetics\cite{Hu1981,M-06}. The formulation of scattering from multiple particles goes back to Foldy's formulation for point scatterers in acoustic scattering \cite{LF45}. Therefore, developing an efficient forward solver for the elastic scattering of multiple particles has its own value beyond the inverse scattering. In this work, we propose a fast and highly accurate numerical method for solving the elastic scattering problem from well-separated multiple particles. The method extends the classic multiple scattering theory for acoustic and electromagnetic waves to elastic waves. It can handle many particles that are arbitrarily shaped and randomly located in a homogeneous medium. The idea goes back to \cite{GG2013, LKG-OE14, LKB-JCP15, LL2019} for the electromagnetic scattering of multiple particles. More specifically, for a given particle, we first construct a scattering matrix that maps the incoming wave to the outgoing wave. If all the particles are identical, up to a shift and rotation, the scattering matrix only has to be computed once. With this matrix precomputed, we then treat the outgoing scattering coefficients, instead of the discretization points on the boundary of particles, as the unknowns in our equation. For a given accuracy, the number of truncated terms in the outgoing scattering coefficients is much less than the number of boundary unknowns, especially for particles with complicated geometry. The resulted system is then solved by an iterative solver accelerated by the fast multipole method (FMM) \cite{R-90}. Numerically we  demonstrate that the algorithm is well suited for the forward simulation of elastic scattering from multiple particles. 

The paper is organized as follows. In section 2, we formulate the  scattering and inverse scattering problems of multiple elastic particles in three dimensions.  Section 3 develops the time reversal method to solve the inverse multiple elastic obstacles problem. Section 4 gives the asymptotic analysis for the selective focusing of small and distant particles. Section 5 presents the multiple scattering theory for elastic particles and proposes the fast algorithm based on the scattering matrix and fast multipole method.  Numerical experiments are presented in section 6 to show that the TRM can effectively solve the  inverse scattering of multiple elastic particles with non-negligible interactions based on the fast solver. The paper is concluded in section 7.  

 \section{Problem formulation}
 
 Consider $M$ elastically rigid particles in three dimensions, denoted by $D_1,D_2,\dots, D_M$. Assume their boundaries $\Gamma_{1}, \Gamma_{2}, \dots, \Gamma_{M}$ are smooth. Let $D=D_{1}\cup D_{2}\cup\cdots\cup D_{M}$ and $\Gamma=\Gamma_{1}\cup\Gamma_{2}\cup\cdots\cup\Gamma_{M}$.  Denote $\nu$ the unit exterior normal vector on
 $\Gamma$. The exterior
 domain $\mathbb{R}^3\setminus \overline{D}$ is assumed to be filled with a homogeneous and isotropic elastic medium with a unit mass density (i.e. $\rho = 1$). Let the particles be illuminated by a time-harmonic incident wave $\boldsymbol{u^i}$.  The displacement of the total field $\boldsymbol u$ consists of the incident field $\boldsymbol{u^i}$ and the scattered field $\boldsymbol v$, i.e., $ \boldsymbol u=\boldsymbol{u^i}+\boldsymbol v$, and satisfies the Navier equation
 \begin{equation}\label{navierequ}
 \mu\Delta\boldsymbol{u}+(\lambda+\mu)\nabla\nabla\cdot\boldsymbol{u}
 +\omega^2\boldsymbol{u}=0,\mbox{ in } \mathbb{R}^3\setminus \overline{D}, 	
 \end{equation}
 where $\omega>0$ is the angular frequency and $\lambda, \mu$ are the
 Lam\'{e} constants satisfying $\mu>0, \lambda+\mu>0$.
 Since we assume the particles are rigid, it holds 
 \begin{eqnarray}\label{rigidboundary}
 	\boldsymbol{u}=0\quad {\rm on}~\Gamma.
 \end{eqnarray}
 
 The incident wave $\boldsymbol{u^i}$ is given as a linear combination of  a longitudinal plane wave $$\boldsymbol{u^i}(x,\alpha,f_{\fp}(\alpha),0):=\boldsymbol{u^i_{\fp}}(x)=f_{\fp}(\alpha) \mathrm{e}^{\mathrm{i} \kappa_{\mathfrak p}\alpha\cdot x}$$ and a transversal plane wave $$\boldsymbol{u^i}(x,\alpha,0,f_{\fs}(\alpha)):=\boldsymbol{u^i_{\fs}}(x)=f_{\fs}(\alpha)
 \mathrm{e}^{\mathrm{i}\kappa_{\mathfrak s} \alpha\cdot x},$$
 where $\alpha=(\sin\theta\cos\varphi, \sin\theta\sin\varphi, \cos\theta)^\top$ is the unit propagation vector,  $f_{\fp}(\alpha)\in \mathbb{C}^3$ is the longitudinal vector satisfying $f_{\fp}(\alpha)\times \alpha = 0$ and $f_{\fs}(\alpha)\in \mathbb{C}^3$ is the transversal vector satisfying $f_{\fs}(\alpha)\cdot \alpha = 0 $. Note that $\boldsymbol{u^i_{\fp}}$ and $\boldsymbol{u^i_{\fs}}$ are also called the compressional and shear incident wave, respectively, with compressional wavenumber $\kp$ and shear wavenumber $\ks$ defined by
 \[
 \kappa_{\mathfrak p}=\frac{\omega}{\sqrt{\lambda+2\mu}},\quad
 \kappa_{\mathfrak s}=\frac{\omega}{\sqrt{\mu}}. 
 \]
 
 It is easy to verify that the scattered field $\boldsymbol v$ satisfies
 the boundary value problem
 \begin{equation}\label{scatteredfield}
 	\begin{cases}
 		\mu\Delta\boldsymbol{v}+(\lambda+\mu)\nabla\nabla\cdot\boldsymbol{v}
 		+\omega^2\boldsymbol{v}=0\quad &{\rm in}~
 		\mathbb{R}^3\setminus\overline{D},\\
 		\boldsymbol{v}=-\boldsymbol{u^i}\quad &{\rm on}~\Gamma.
 	\end{cases}
 \end{equation}
 In addition, the scattered field $\boldsymbol v$ is required to satisfy
 the Kupradze--Sommerfeld radiation condition
 \[
 \lim_{r\to\infty}r(\partial_r\boldsymbol v_{\mathfrak p}-\mathrm{i}\kappa_{\mathfrak p}\boldsymbol v_{\mathfrak p})=0,\quad
 \lim_{r\to\infty}r(\partial_r\boldsymbol v_{\mathfrak s}-\mathrm{i}\kappa_{\mathfrak s}\boldsymbol v_{\mathfrak s})=0,\quad r=|x|,
 \]
 where
 \[
 \boldsymbol v_{\mathfrak p}=-\frac{1}{\kappa_{\mathfrak p}^2}\nabla\nabla\cdot\boldsymbol v,\quad
  \boldsymbol v_{\mathfrak s}=\frac{1}{\kappa_{\mathfrak s}^2}{\bf curl curl}\boldsymbol v, 
 \]
 are known as the  compressional and shear wave components of $\boldsymbol v$,
 respectively. 
 
The fundamental solution of the Navier equation \eqref{navierequ} in the free space of three dimensions is given by
 \begin{eqnarray}\label{greenfun}
 	\Phi(x,y) = \frac{\ks^2}{4\pi\omega^2}\frac{e^{\mi \ks|x-y| }}{|x-y|}I+\frac{1}{4\pi \omega^2}\nabla\nabla^\top \left[\frac{e^{\mi \ks|x-y| }}{|x-y|}-\frac{e^{\mi \kp|x-y| }}{|x-y|}\right],
 \end{eqnarray}
where $I$ is the $3\times 3$ identity matrix. The traction operator $T_\nu$ on $\Gamma$ is defined by 
\begin{eqnarray}
	T_\nu: = 2\mu \nu\cdot \nabla + \lambda \nu\nabla\cdot + \mu\nu\times {\bf curl}.
\end{eqnarray}
Based on the Betti's formula \cite{AK2002} and the boundary condition \eqref{rigidboundary}, we can represent the scattered field $\boldsymbol v$ through the single layer integral formulation, 
\begin{eqnarray}\label{singleintrep}
	\bv(x) =-\int_{\Gamma} \Phi(x,y)T_{\nu(y)}\bu(y) ds_y, \quad x\in 	\mathbb{R}^3\setminus\overline{D}.
\end{eqnarray}
As $|x|\rightarrow\infty$, the asymptotic behavior of the elastic scattered field $\bv$ is given by 
\begin{equation}\label{farfield}
\bv(x) = \frac{e^{\mi \kp |x|}}{|x|} \bv_{\fp,\infty}(\hat{x})+ \frac{e^{\mi \ks |x|}}{|x|} \bv_{\fs,\infty}(\hat{x}) +\mathcal{O}\left(\frac{1}{|x|^2}\right),
\end{equation}
where   $\bv_{\fp,\infty}$ and $\bv_{\fs,\infty}$ are defined on the unit sphere $\mathbb{S}^2$ with $\hat{x}=x/|x|$ and known as the compressional and shear wave far field pattern, respectively.  Based on the asymptotic behavior of the fundamental solution \eqref{greenfun}, it can be verified that
\begin{eqnarray}
	\begin{split}
	\bv_{\fp,\infty}(\hat{x}) &= -\frac{\kp^2}{4\pi\omega^2}\int_{\Gamma}\hat{x}\hat{x}^\top e^{-\mi\kp \hat{x}\cdot y} T_{\nu(y)} \bu(y) ds_y,\\
	\bv_{\fs,\infty}(\hat{x}) &= -\frac{\ks^2}{4\pi\omega^2}\int_{\Gamma}[I-\hat{x}\hat{x}^\top] e^{-\mi\ks \hat{x}\cdot y} T_{\nu(y)} \bu(y) ds_y.
\end{split} 
\end{eqnarray}
The forward  problem for the elastic scattering of multiple particles is: 
\begin{itemize}
\item Given the incident wave $\uinc$ and the geometry information of $D_1, \cdots, D_M$, find the far field pattern $\bv_{\fp,\infty}(\hat{x})$ and $\bv_{\fs,\infty}(\hat{x})$ of the scattered field $\bv$. 
\end{itemize}
The inverse problem of elastic scattering of multiple particles is:
\begin{itemize}
 \item Based on the far field pattern $\bv_{\fp,\infty}(\hat{x})$ and $\bv_{\fs,\infty}(\hat{x})$ from different incident directions,  recover the geometry information of $D_1, D_2, \cdots, D_M$, including locations and shapes. 
\end{itemize}
It is worth mentioning that the two problems are equally important. They are connecting to each other in the sense that solving the inverse problem often requires solving the forward problem many times, including both the direct and iterative type inversion methods. Many algorithms concerning the forward and inverse problems have been discussed in the literature\cite{BXY2017,BLR2014,DLL2021,LY2019}, and yet fast algorithms for the elastic scattering and inverse scattering from many particles in three dimensions are still very rare.  In the following sections, we will first introduce the time reversal method (TRM) to  the inverse elastic scattering problem and then discuss the fast algorithm for elastic scattering of multiple particles.

\section{Inverse elastic scattering based on TRM}
Denote $\bv_{\fp,\infty}(\hat{x},\alpha,f_{\fp}(\alpha),f_{\fs}(\alpha))$, $\bv_{\fs,\infty}(\hat{x},\alpha,f_{\fp}(\alpha),f_{\fs}(\alpha))$ the compressional  and shear wave far field patterns radiated by the incident wave  $\boldsymbol{u^i}(x,\alpha,f_{\fp}(\alpha),f_{\fs}(\alpha)) = \boldsymbol{u^i}(x,\alpha,f_{\fp}(\alpha),0)+\boldsymbol{u^i}(x,\alpha,0,f_{\fs}(\alpha))$. Define the $L^2$ space 
\begin{eqnarray}
	L^2_{\fp} = \{f_{\fp}: \ms \rightarrow \mathbb{C}^3\ |\  f_{\fp}(\alpha)\times \alpha = 0 , |f_{\fp}|\in L^2(\ms)\}
\end{eqnarray}
of the longitudinal vector fields on $\ms$ and the $L^2$ space 
\begin{eqnarray}
	L^2_{\fs} = \{f_{\fs}: \ms \rightarrow \mathbb{C}^3\ |\ f_{\fs}(\alpha)\cdot \alpha =0, |f_{\fs}|\in L^2(\ms)\}
\end{eqnarray}
of the transversal vector fields, where $|\cdot|$ is the Euclidean norm in $\mathbb{C}^3$. The scalar product on the space $L^2_{\fp}(\ms)\times L^2_{\fs}(\ms)$ is defined by
\begin{eqnarray}\label{innerprod}
	(f,g) = \frac{\omega}{\kp}\int_{\ms} f_{\fp}(\alpha)\cdot \overline{g_{\fp}(\alpha)}ds_\alpha + \frac{\omega}{\ks}\int_{\ms}f_{\fs}(\alpha)\cdot \overline{g_{\fs}(\alpha)}ds_\alpha,
\end{eqnarray}
with $f=(f_{\fp}, f_{\fs})$ and $g = (g_{\fp},g_{\fs})$. The elastic Herglotz wave with kernel $f\in L^2_{\fp}(\ms)\times L^2_{\fs}(\ms) $ has the form
\begin{eqnarray}\label{herglotz}
	\uincf(x) =\int_{\ms} e^{\mi \kp \alpha\cdot x} f_{\fp}(\alpha) + e^{\mi \ks\alpha \cdot x} f_{\fs}(\alpha)ds_{\alpha}, 
\end{eqnarray}
which is a superposition of plane waves and satisfies the Navier equation entirely. By linearity, the corresponding far field operator $$F:L^2_{\fp}(\ms) \times L^2_{\fs}(\ms)  \rightarrow L^2_{\fp}(\ms) \times L^2_{\fs}(\ms)$$  due to the incident wave $\uincf(x)$ is defined by
\begin{eqnarray}
	F(f) = \frac{1}{\omega}\int_{\ms}\bv_{\infty}(\hat{x},\alpha,f_{\fp}(\alpha),f_{\fs}(\alpha))ds_\alpha,
\end{eqnarray}
where $\bv_{\infty}(\hat{x},\alpha,f_{\fp}(\alpha),f_{\fs}(\alpha)) =\left(\bv_{\fp,\infty}(\hat{x},\alpha,f_{\fp}(\alpha),f_{\fs}(\alpha)), \bv_{\fs,\infty}(\hat{x},\alpha,f_{\fp}(\alpha),f_{\fs}(\alpha))\right)$.
It is easy to see that the far field operator $F$ is compact since the kernel is smooth. Using the reciprocity relation of elastic wave, one can show the following result\cite{AK2002}.
\begin{theorem}
	The far field operator $F: L^2_{\fp}(\ms) \times L^2_{\fs}(\ms)  \rightarrow L^2_{\fp}(\ms) \times L^2_{\fs}(\ms)$ is a compact and normal operator. Its adjoint operator $F^*: L^2_{\fp}(\ms) \times L^2_{\fs}(\ms)  \rightarrow L^2_{\fp}(\ms) \times L^2_{\fs}(\ms)$ with respect to the inner product \eqref{innerprod} is given by
	$$F^*f = \overline{RFR\overline{f}}, \quad \forall f \in L^2_{\fp}(\ms) \times L^2_{\fs}(\ms),$$
	where $R$ is the symmetry operator defined by $Rf(\alpha)=f(-\alpha), \alpha\in \ms$.
\end{theorem}
It is worth mentioning that similar result also holds for the acoustic and electromagnetic scattering\cite{DR-book2}. We are now able to define the time reversal operator $T$.  First let us measure the far field of the scattered field due to the Herglotz wave $\uinc_f$ with $f\in L^2_{\fp}(\ms) \times L^2_{\fs}(\ms) $, and then use the conjugate of the far field as the kernel $g$ of a new Herglotz wave. In other words,
\begin{eqnarray*}
	g = \overline{RFf}.
\end{eqnarray*}   
The symmetry operator $R$ is used here in order to reemit the wave from the opposite of the measured direction. The time reversal operator $T$ is then obtained by iterating this cycle twice
\begin{eqnarray}
	T f = \overline{RF g} = \overline{RF \overline{RFf}}.
\end{eqnarray} 
It holds the following property for the time reversal operator $T$.
\begin{theorem}
	The time reversal operator $T$ is compact, self-adjoint and positive. It is defined as an operator from $L^2_{\fp}(\ms) \times L^2_{\fs}(\ms) $ to itself with
	\begin{eqnarray}
		Tf = FF^*f = F^*F f 
	\end{eqnarray}
	The nonzero eigenvalues of $T$ are exactly positive numbers
	$|\lambda_1|^2\ge |\lambda_2|^2 \ge \cdots >0$
	where the sequence $(\lambda_j)_{j\ge 1}$ denotes the nonzero complex eigenvalue of the far field operator $F$. The corresponding eigenfunctions $(f_j)_{j\ge 1}$ of $F$ are exactly the eigenfunctions of $T$. If non-trivial solutions to the Navier equation in $D$ with homogeneous Dirichlet do not exist, then $(f_j)_{j\ge 1}$ form a complete orthonormal system in $L^2_{\fp}(\ms) \times L^2_{\fs}(\ms) $. 
\end{theorem}
This is also shown in \cite{AK2002}, which shows that all the eigenvalues of the far-field operator  $F$ lie on the circle with center at $(0,\pi/\omega)$ on the positive imaginary axis and radius $2\pi/\omega$. The time reversal method is to illuminate an obstacle with Herglotz waves with kernel $f$ corresponding to an eigenvector of $F$ (or $T$) with non-zero eigenvalue. In particular, the Herglotz wave generated by $f$ with $\lambda\ne 0$ will automatically focus on the obstacles, as shown by the following theorem.
\begin{theorem}\label{globalf}
	Let $\lambda\ne 0$ be an eigenvalue of $F$ and $f\in  L^2_{\fp}(\ms) \times L^2_{\fs}(\ms)$ an eigenvector of $F$ associated with $\lambda$. Then, the Herglotz wave $\uincf$ associated with $f$ has the following form
	\begin{eqnarray*}
		\uincf =  -\frac{\kp^2}{ \lambda\omega^3} \int_{\Gamma_D}\hat{x}\hat{x}^\top j_0(\kp||x-y||) T_\nu(y)\bu_{f} ds_y -\frac{\ks^2}{ \lambda\omega^3} \int_{\Gamma_D}(I-\hat{x}\hat{x}^\top)j_0(\ks||x-y||) T_\nu(y)\bu_{f} ds_y
	\end{eqnarray*}
\end{theorem}
\begin{proof}
	Since $f=(f_{\fp},f_{\fs})\in  L^2_{\fp}(\ms) \times L^2_{\fs}(\ms)$ is an eigenvector of $F$ with eigenvalue $\lambda\ne0$, it holds
	\begin{eqnarray}
		f_{\fp}(\hat{x}) & = & \frac{1}{\lambda\omega}\int_{\ms}\bv_{\fp,\infty}(\hat{x},\alpha,f_{\fp}(\alpha),f_{\fs}(\alpha))ds_\alpha \notag \\
		& = & -\frac{\kp^2}{4\pi \lambda\omega^3}\int_{\ms} \int_{\Gamma_D}\hat{x}\hat{x}^\top e^{-\mi\kp \hat{x}\cdot y} T_{\nu(y)} \bu(y,\alpha,f_{\fp}(\alpha),f_{\fs}(\alpha)) ds_yds_\alpha \notag \\
		&=& -\frac{\kp^2}{4\pi \lambda\omega^3}\int_{\Gamma_D}\hat{x}\hat{x}^\top e^{-\mi\kp \hat{x}\cdot y} \int_{\ms}  T_{\nu(y)} \bu(y,\alpha,f_{\fp}(\alpha),f_{\fs}(\alpha)) ds_\alpha ds_y\notag \\
		&=& -\frac{\kp^2}{4\pi \lambda\omega^3}\int_{\Gamma_D}\hat{x}\hat{x}^\top e^{-\mi\kp \hat{x}\cdot y} T_\nu(y)\bu_{f} ds_y
	\end{eqnarray}
	where $T_\nu(y)\bu_{f}=\int_{\ms}  T_{\nu(y)} \bu(y,\alpha,f_{\fp}(\alpha),f_{\fs}(\alpha)) ds_\alpha$ is the traction of the elastic field generated by the Herglotz wave $\uincf$. 
	
	Similarly, it holds
	\begin{eqnarray}
		f_{\fs} (\hat{x})  = -\frac{\ks^2}{4\pi \lambda\omega^3}\int_{\Gamma_D}[I-\hat{x}\hat{x}^\top ] e^{-\mi\ks \hat{x}\cdot y} T_\nu(y)\bu_{f} ds_y.
	\end{eqnarray} 
	
	Now plugging $f$ into the definition of Herglotz wave \eqref{herglotz}, we obtain
	\begin{eqnarray}
		\uincf &=& \int_{\ms} e^{\mi \kp \alpha\cdot x} f_{\fp}(\alpha) + e^{\mi \ks\alpha \cdot x} f_{\fs}(\alpha)ds_{\alpha}\notag \\
		&=& -\frac{\kp^2}{4\pi \lambda\omega^3} \int_{\Gamma_D}\hat{x}\hat{x}^\top\left(\int_{\ms}  e^{\mi \kp \alpha\cdot x}  e^{-\mi\kp \alpha \cdot y} ds_\alpha  \right) T_\nu(y)\bu_{f} ds_y \notag \\&& -\frac{\ks^2}{4\pi \lambda\omega^3} \int_{\Gamma_D}(I-\hat{x}\hat{x}^\top)\left(\int_{\ms}  e^{\mi \ks \alpha\cdot x}  e^{-\mi\ks \alpha \cdot y} ds_\alpha  \right) T_\nu(y)\bu_{f} ds_y. \notag
	\end{eqnarray} 
	The conclusion follows from  the identity 
	\[
	\int_{\ms}  e^{\mi \kappa \alpha\cdot x}  e^{-\mi\kappa \alpha \cdot y} ds_\alpha = 4\pi j_0(\kappa||x-y||), \mbox{ for }\kappa = \kp\mbox{ or } \ks,
	\]
	which completes the proof.
\end{proof}
Based on the property of $j_0(r)$, the incident wave $\uincf$ generated by the eigenfunction $f$ will focus on the unknown obstacles and decay as $1/r$ where $r$ is the distance from the obstacle, which is the essential property of the time reversal method. Theorem \ref{globalf} also shows one can use only one wave (either compressional or shear wave) to focus the obstacles. Meanwhile, since the time reversal operator $T$ is self-adjoint, by min-max principle, it holds that
$$|\lambda_1|^2=\sup_{f\in L^2_{\fp}(\ms) \times L^2_{\fs}(\ms),||f||^2_2=1}||Ff||^2_2.$$
Therefore, the eigenfunction of the largest eigenvalue will maximize the illumination of particles. In general, for particles with non-negligible interactions, it is difficult to obtain the explicit form of significant eigenvalues, as well as the eigenfunctions, for the far field operator $F$ with respect to the locations of particles. Only global focusing can be achieved for the TRM applied to the inverse scattering of multiple particles that are reasonably dense. However, for small and distant particles, selective focusing\cite{CK04} an individual particle can be realized when the interaction among particles becomes weak, as shown in the next section.

\section{Focusing of small particles}
In this section, we assume the particles are far from each other so that multiple scattering among different particles is negligible. In this case, the inverse scattering of multiple particles is essentially reduced to the reconstruction of a single particle, which is also called inverse ``independent scattering'' in the literature\cite{Hu1981}. Here we first investigate the property of far field operator of a small sphere and then discuss the selective focusing of general-shaped small elastic particles.   

\subsection{Focusing of a small sphere}
The elastic scattering of a sphere can be fully characterized by Mie theory\cite{Louer2014}. Denote $S_0$ the sphere centered at the origin with radius $R$. For a given point $x = (x_1,x_2,x_3)$, denote $(r,\theta,\phi)$ the spherical coordinates of $x$. Let $Y_n^m$, $m=-n,\dots,n$, $n=0,1,\dots$ be the orthonormal spherical harmonics on the unit sphere $\ms$, which are defined by
\begin{eqnarray*}
	Y_n^m(\theta,\phi) = (-1)^m\sqrt{\frac{2n+1}{4\pi}\frac{(n-|m|)!}{(n+|m|)!}}P_n^{|m|}(\cos\theta)e^{im\phi},
\end{eqnarray*}
where $P_n^{|m|}(t)$, $t\in[-1,1]$, are the associated Legendre functions\cite{DR-book2}. Denote ${\rm Grad}Y_n^m$ the surface gradient of $Y_n^m$ on $\ms$. Let $j_n(r)$ be the spherical Bessel function and $h_n^{(1)}(r)$ be the first kind spherical Hankel function of order $n$. Define the scalar functions
\begin{eqnarray}
	u_{n,m}^\kappa (x) = j_n(\kappa r)Y_n^m(\theta,\phi), \quad  v_{n,m}^\kappa (x) = h^{(1)}_n(\kappa r)Y_n^m(\theta,\phi),
\end{eqnarray}
which are \textit{spherical wave functions} and satisfy the three dimensional Helmholtz equation with exceptional point at the origin for $v^{\kappa}_{n,m}(x)$.
Based on the property of spherical harmonics, the incoming elastic field for the sphere $S_0$ can be expanded as\cite{Louer2014}
\begin{eqnarray}\label{inspan}
	\uinc(x) = \sum_{n=1}^{\infty}\sum_{m=-n}^{n}\left(a_{n,m}\nabla\times\nabla\times(x u_{n,m}^{\ks})/(\mi\ks) +b_{n,m}\nabla\times(x u_{n,m}^{\ks}) \right) + \sum_{n=0}^{\infty}\sum_{m=-n}^{n}c_{n,m}\nabla u_{n,m}^{\kp},
\end{eqnarray}
where $\{a_{n,m},b_{n,m},c_{n,m}\}$ are called the \textit{incoming expansion coefficients} of $\uinc$ on $S_0$. Note that the shear part $\uinc_{\fs}$ and the compressional part $\uinc_{\fp}$ of an incident wave $\uinc$ in \eqref{inspan} are
\begin{eqnarray*}
	\uinc_{\fs}&=&\frac{1}{\kappa_{\fs}^2}{\bf curl curl}\uinc=
	\sum_{n=1}^{\infty}\sum_{m=-n}^{n}\left(a_{n,m}\nabla\times\nabla\times(x u_{n,m}^{\ks})/(\mi\ks) +b_{n,m}\nabla\times(x u_{n,m}^{\ks}) \right),\\ \uinc_{\fp}&=&-\frac{1}{\kappa_{\fp}^2}\nabla\nabla\cdot\uinc= \sum_{n=0}^{\infty}\sum_{m=-n}^{n}c_{n,m}\nabla u_{n,m}^{\kp}.
\end{eqnarray*}
\begin{remark}
	To eliminate the difference of starting indices between the expansions of the shear and compressional incident wave, we simply let $a_{n,m}=b_{n,m}=0$ for $n=0$ and write all the expansion coefficients starting from $n=0$ uniformly. This applies to all the related quantities based on the spherical harmonic expansion unless otherwise stated.
\end{remark}

For the plane wave incidence, explicit expression for these coefficients can be obtained through the vector analogue of the Jacobi–Anger expansion\cite{DR-book2}. Detailed expression is given in the appendix \ref{appA}.
The kernel $f\in L^2_{\fp}(\ms)\times L^2_{\fs}(\ms)$ of a Herglotz wave can be expanded by
\[f=\sum_{n=0}^{\infty}\sum_{m=-n}^{n}\left(f^{a}_{n,m}{\rm Grad}Y_{n}^m(\hat{x})+f^{b}_{n,m}\hat{x}\times{\rm Grad}Y_{n}^m(\hat{x})+f^c_{n,m}Y_{n}^m(\hat{x})\hat{x}\right).\]
Based on the plane wave expansion \eqref{planespan}, the \textit{incoming expansion coefficients} for a Herglotz wave $\uincf$ are simply
\begin{eqnarray}\label{hegexpan}
	a_{n,m}=4\pi \mi^n f^a_{n,m},\quad b_{n,m}=4\pi \mi^n f^b_{n,m},\quad c_{n,m} = -4\pi \mi^{n+1} f^c_{n,m}/\kp.
\end{eqnarray}
After the incidence of $\uinc$, the scattered field $\bv$ in the exterior of $S_0$ is given by
\begin{eqnarray}\label{scatspan}
	\bv = \sum_{n=0}^{\infty}\sum_{m=-n}^{n}\left(\alpha_{n,m}\nabla\times\nabla\times(x v_{n,m}^{\ks})/(\mi\ks) +\beta_{n,m}\nabla\times(x v_{n,m}^{\ks})   + \gamma_{n,m}\nabla v_{n,m}^{\kp}\right),
\end{eqnarray} 
where $\{\alpha_{n,m},\beta_{n,m},\gamma_{n,m}\}$ are referred as the \textit{outgoing expansion coefficients}. The linear relation from the \textit{incoming expansion coefficients} $\{a_{n,m},b_{n,m},c_{n,m}\}$ to the \textit{outgoing expansion coefficients} $\{\alpha_{n,m},\beta_{n,m},\gamma_{n,m}\}$ is defined as the scattering matrix $\mathcal{S}$. For the sphere $S_0$, the scattering matrix $\mathcal{S}$ is block diagonal with diagonal block $\mathcal{S}_{n,m}$, and the explicit expression is given in the appendix \ref{appB}. 

The far field pattern for the scattered field $\bv$ based on the \textit{outgoing expansion coefficients} is given by
\begin{eqnarray}\label{farexpan}
	\begin{split}
\bv_{\fs,\infty} &= \sum_{n=0}^{\infty} \sum_{m=-n}^{n} \frac{(-\mi)^{n+1}}{\ks} \left( \alpha_{n,m}{\rm Grad}Y_{n}^m(\hat{x}) +\beta_{n,m} \hat{x} {\rm Grad}Y_{n}^m(\hat{x}) \right),  \\
\bv_{\fp,\infty} &= \sum_{n=0}^{\infty} \sum_{m=-n}^{n} (-\mi)^n \gamma_{n,m} Y_{n}^m(\hat{x})\hat{x}.
\end{split}
\end{eqnarray}
By using the scattering matrix of the sphere $S_0$, the far field operator $F$ for  $S_0$ can be formulated as a block diagonal matrix, where the $nm$-th block $F_{n,m}$ is 
\begin{eqnarray}
	F_{n,m} = D^{scat}_{n,m}\mathcal{S}_{n,m} D^{inc}_{n,m}, m = -n,\cdots,n, n = 0,1,\cdots.
\end{eqnarray} 
Here $D^{inc}_{n,m}$ and $D^{scat}_{n,m}$ are $3\times3$ diagonal matrices with the diagonal elements given by equations \eqref{hegexpan} and \eqref{farexpan}.   If the radius $R$ of the sphere $S_0$ is sufficiently small, we are able to obtain the following result.
\begin{theorem}\label{smalleig}
When $n\ge 1$ and $R\rightarrow 0 $, the three eigenvalues of $F_{n,m}$, given by $\lambda^{i}_{n,m}, i = 1,2,3$, satisfy
	\begin{eqnarray}
		\begin{split}
		\lambda^{1}_{n,m} &= \frac{c_n\mi(2n-1)(2n+1)(\ks^{2n}(n+1)+\kp^{2n}n)}{(n+1)\ks^2+n\kp^2}R^{2n-1}+\mathcal{O}(R^{2n+1}),\\
			\lambda^{2}_{n,m} &= \frac{4\pi \mi  j_n(\ks R)}{\ks h_n^{(1)}(\ks R)}=c_n\mi\ks^{2n}R^{2n+1}+\mathcal{O}(R^{2n+3}),\\
		\lambda^{3}_{n,m} &= \frac{c_n\mi(n\ks^2+(n+1)\kp^2)\kp^{2n}\ks^{2n}}{\ks(2n+3)(\ks^{2n}(n+1)+\kp^{2n}n)(1+2n)}R^{2n+3}+\mathcal{O}(R^{2n+5}),
	\end{split}
	\end{eqnarray}
where $c_n = \frac{4\pi^2 \mi }{2^{2n+1}\Gamma(n+1/2)\Gamma(n+3/2)}$. The corresponding eigenfunction for $\lambda^2_{n,m}$ is $\hat{x}\times {\rm Grad}Y_n^m(\hat{x})$, and the eigenfunctions for $\lambda^1_{n,m}$ and $\lambda^3_{n,m}$ lie in the space of $\{{\rm Grad}Y_n^m(\hat{x}),Y_n^m(\hat{x})\hat{x}\}$. When $n=0$, the only eigenvalue is $\lambda_{0,0}=\frac{4\pi\mi j'_0(\kp R)}{\kp h^{(1)'}_0(\kp R)}= -\frac{4\pi}{3} \kp^2 R^3+\mathcal{O}(R^5)$ with eigenfunction $Y_0^0(\hat{x})\hat{x}$. 
\end{theorem}
 Proof is straightforward calculation based on the explicit form of $F_{n,m}$ and the asymptotic expansions of Bessel functions\cite{Hand2010} for $z\rightarrow 0$ 
 \begin{eqnarray}\label{asybessel}
 	\begin{split}
 		j_n(z) &= \sqrt{\pi}/(2^{n+1}\Gamma(n+3/2))\left(z^n-z^{n+2}/(6+4n)\right)+O(z^{n+4}), \\
 		h_n^{(1)}(z) & =  -\mi 2^n\Gamma(n+1/2)/\sqrt{\pi}(z^{-(n+1)}+z^{-(n-1)}/(4n-2)) + O(z^{-n+3}).
 	\end{split}
 \end{eqnarray}
 From Theorem \ref{smalleig}, one can see that the eigenvalues decrease as $R^2$ as $n$ increases, which is similar to the behavior of acoustic scattering\cite{CK04}. In particular, for a small sphere, there are three significant eigenvalues ($\lambda^1_{1,m}$, $m=-1,0,1$),  whose eigenfunctions will dominate the far field scattering. Such a conclusion can be extended to the general-shaped small particles and one can use it to achieve selective focusing. 
\subsection{Selective focusing}
Consider a family of particles $\{D^{\varepsilon}_l, l = 1,\cdots, M\}$. Each $D^{\varepsilon}_l$ is obtained from a reference domain $D_l$ by a dilation ratio $\varepsilon$ and a translation
\begin{eqnarray}\label{dilation}
	D^{\varepsilon}_l =\left\{x=s_l+\varepsilon \xi; \xi\in D_l \right\},
\end{eqnarray}
where $s_l$ is the center of $D^{\varepsilon}_l$.  We choose $s_l$ differently so that all the obstacles are not intersecting with small $\varepsilon$. 

Let $\Gamma^{\varepsilon}_l$ denote the boundary of $D^{\varepsilon}_l$. 
In order to study the asymptotic behavior of the elastic scattering, we let $\kappa =\max\{\kp,\ks\}$ and denote $\Phi_{\kappa}(x,y)$ the fundamental solution \eqref{greenfun} that depends on $\kappa$. Based on the integral representation \eqref{singleintrep}, the scattered field can be represented by
\begin{eqnarray}
	\bv(x)  =  \sum_{j=1}^M \int_{\Gamma^{\varepsilon}_j}\Phi_\kappa(x,y) \bJ_j(y) ds_y, \quad x\in \Omega = \mathbb{R}^3\backslash \cup_{j = 1}^M D^{\varepsilon}_j, 
\end{eqnarray}  
where $\bJ_j$ is the unknown traction field on $\Gamma^{\varepsilon}_j$ (up to a sign). Since all the obstacles are rigid, given the Herglotz incident wave $\uincf$,  it holds
\begin{eqnarray}\label{boundcond}
	\bv|_{\Gamma^{\varepsilon}_l} = -\uincf|_{\Gamma^{\varepsilon}_l}, \quad  l =  1,\cdots,M.
\end{eqnarray}
Denote the single layer boundary operator as $$\mathcal{S}^{\kappa,\varepsilon}_{l,j}\bJ_j= \int_{\Gamma^{\varepsilon}_j}\Phi_\kappa(x,y) \bJ_j(y) ds_y, \quad x\in \Gamma^{\varepsilon}_l.$$ Based on the boundary condition \eqref{boundcond}, we obtain the boundary integral system
\begin{eqnarray}\label{smallcoupint}
	\sum_{j=1}^M \mathcal{S}^{\kappa,\varepsilon}_{l,j}\bJ_j = -\uincf|_{\Gamma^{\varepsilon}_l}, \quad  l = 1,\cdots,M.
\end{eqnarray}

In order to analyze the solution of \eqref{smallcoupint} when  $\varepsilon\rightarrow 0$, we make the change of variables
\begin{eqnarray*}
	\xi =  \frac{x-s_l}{\varepsilon}\in D_l, \quad
	\eta =  \frac{y-s_j}{\varepsilon}\in D_j,
\end{eqnarray*}
so the scaled vector fields and fundamental solution become
\begin{eqnarray*}
	\bJ^\varepsilon_j (\eta) &=& \bJ_j(x),\\ 
	\Phi_{l,j}^{\kappa,\varepsilon}(\xi,\eta) &=& \Phi_\kappa(x,y),
\end{eqnarray*}
for $x\in D_l, y\in D_j$. Based on the rescaled variables, we have 
\begin{eqnarray}\label{rescalsing}
	S^{\kappa,\varepsilon}_{l,j}\bJ^\varepsilon_j = \varepsilon^2 \int_{\Gamma_j}\Phi_{l,j}^{\kappa,\varepsilon}(\xi,\eta)\bJ^\varepsilon_j (\eta)ds_\eta.
\end{eqnarray}
The boundary integral system becomes
\begin{eqnarray}\label{rescaledequ}
	\sum_{j=1}^M S^{\kappa,\varepsilon}_{l,j}\bJ^\varepsilon_j = -\boldsymbol{u}^{\boldsymbol{i},\varepsilon}_f|_{\Gamma_l}, \quad  l = 1,\cdots,M,
\end{eqnarray}
with $\boldsymbol{u}^{\boldsymbol{i},\varepsilon}_f(\xi) = \uincf(x)$. Now consider the diagonal term  in \eqref{rescaledequ} and let $\varepsilon\rightarrow 0$. Using the fact that
\begin{eqnarray}
	\Phi_{l,l}^{\kappa,\varepsilon}(\xi,\eta) = \frac{1}{\varepsilon}\Phi_{\kappa\varepsilon}(\xi,\eta), 
\end{eqnarray} 
we obtain that the diagonal term is 
\begin{eqnarray}
	S^{\kappa,\varepsilon}_{l,l}\bJ^{\varepsilon}_l = \varepsilon \int_{\Gamma_l}\Phi_{\kappa\varepsilon}(\xi,\eta)\bJ^\varepsilon_l (\eta)ds_\eta = \varepsilon \tilde{S}^{\kappa,\varepsilon}_{l,l} \bJ^{\varepsilon}_l,\quad  l = 1,\cdots,N.
\end{eqnarray}
When $\kappa\varepsilon\rightarrow 0$, $\tilde{S}^{\kappa,\varepsilon}_{l,l}$ converges to the single layer boundary integral operator of the zero frequency problem, namely, the Lam\'e system. In particular, let $\bv_{i}(\xi)$, $i=1,2,3$, satisfy the Lam\'e system
\begin{eqnarray}\label{lame}
	\begin{cases}
		\mu\Delta\boldsymbol{v}_i+(\lambda+\mu)\boldsymbol{\nabla}\boldsymbol{\nabla}\cdot\boldsymbol{v}_i =0\quad &{\rm in}~
		\mathbb{R}^3\setminus\overline{\mathcal{O}_l},\\
		\boldsymbol{v}_i=-\boldsymbol{e}_i\quad &{\rm on}~\Gamma_l, \\
		\bv_i(\xi) = O\left(\frac{1}{|\xi|}\right),\quad \nabla \bv_i(\xi) = o\left(\frac{1}{|\xi|}\right), \quad &\mbox{as } |\xi|\rightarrow +\infty.
	\end{cases}
\end{eqnarray} 
where $\{\boldsymbol{e}_i\}_{i=1}^3$ is the basis vector of $\mathbb{R}^3$. Then the boundary integral equation
	\begin{eqnarray}\label{zerobdint}
	\int_{\Gamma_l}\Phi_{0}(\xi,\eta)\bJ_i^0(\eta)ds_\eta = -\boldsymbol{e}_i,
	\end{eqnarray}
admits a unique solution for $i = 1,2,3$ with $\bJ_i^0(\eta)=-T_{\nu}\bv_i.$
Since $\Phi_{\kappa\varepsilon}(\xi,\eta)=\Phi_0(\xi,\eta) + \mathcal{O}(\kappa\varepsilon)$ as $\kappa\varepsilon\rightarrow 0$, the diagonal operator $S^{\kappa,\varepsilon}_{l,l}$ is invertible when $\kappa\varepsilon$ is sufficiently small. In addition, the solution to the equation
\begin{eqnarray}
	S^{\kappa,\varepsilon}_{l,l}\bJ^{\varepsilon}_l = -\boldsymbol{u}^{\boldsymbol{i},\varepsilon}_f|_{\Gamma_l}, \quad  l = 1,\cdots,M,
\end{eqnarray}
can be approximated by 
\begin{eqnarray}
	\varepsilon\bJ_l^\varepsilon  = -\sum_{i=1}^3(\boldsymbol{u}^{\boldsymbol{i},\varepsilon}_f\cdot\boldsymbol{e}_i)T_{\nu}\bv_i|_{\Gamma_l}  +\mathcal{O}(\kappa\varepsilon).
\end{eqnarray}
Plugging into \eqref{rescalsing} yields
\begin{eqnarray}\label{approxfield}
	S^{\kappa,\varepsilon}_{l,l}\bJ^\varepsilon_l = -\varepsilon \int_{\Gamma_l}\Phi_{l,l}^{\kappa,\varepsilon}(\xi,\eta) \left( \sum_{i=1}^3(\boldsymbol{u}^{\boldsymbol{i},\varepsilon}_f\cdot\boldsymbol{e}_i)T_{\nu}\bv_i|_{\Gamma_l}  \right)ds_\eta  +\mathcal{O}(\kappa\varepsilon^2).
\end{eqnarray}
For the off diagonal terms $S^{\kappa,\varepsilon}_{l,j}$ with $l\ne j$, let
\begin{eqnarray}\label{dist}
	d = \min_{1\le l\ne j\le M}|s_l-s_j|,
\end{eqnarray}
be the minimal distance  between the centers of particles. Using the relation
\[|s_l-s_j+\varepsilon(\xi-\eta)|=|s_l-s_j|\left(1+\mathcal{O}(\varepsilon/d)\right),\]
one can see that 
\[\Phi_{l,j}^{\kappa,\varepsilon}= \Phi_{\kappa}(s_l,s_j)\left[1+\mathcal{O}(\kappa\varepsilon)\right],\]
which implies 
\begin{eqnarray}
	||S^{\kappa,\varepsilon}_{l,j}\bJ^{\varepsilon}_j||_\infty = \mathcal{O}(\varepsilon/d)\left[1+\mathcal{O}(\kappa\varepsilon)\right] , \mbox{ for } l\ne j,
\end{eqnarray}
where $||\cdot||_\infty$ is the maximum norm for continuous functions. It shows the multiple scattering effect can be neglected and the elastic field can be obtained by the superposition of low frequency scattering problem if the distance $d$ is large enough. Based on the far field pattern of equation \eqref{approxfield}, we obtain the following theorem
\begin{theorem}\label{multiscatt}
	Assume there are $M$ elastic particles given by \eqref{dilation} with size $\varepsilon$ sufficiently small and the distance $d$ between any two particles sufficiently large. For all $l = 1,2,\cdots, M$ and $i=1,2,3$,  let $\bv_{i,l}$ satisfy the Lam\'e system \eqref{lame}. Then the far field pattern $\bv_{\infty}(\hat{x},\alpha,f_{\fp}(\alpha),f_{\fs}(\alpha))$ of the scattering problem \eqref{scatteredfield} for the $M$ particles  is given by
	\begin{eqnarray}\label{farapprox}
		\begin{split}
		\frac{1}{\varepsilon}\bv_{\fp,\infty}(\hat{x},\alpha,f_{\fp}(\alpha),f_{\fs}(\alpha))&= \frac{-\kp^2}{4\pi\omega^2} \sum_{l=1}^M \hat{x}\hat{x}^T  e^{-\mi\kp\hat{x}\cdot s_l} \bv^{l}_{\infty}(\alpha,f_{\fp}(\alpha),f_{\fs}(\alpha)) +\mathcal{O}(\kappa\varepsilon),  \\
		\frac{1}{\varepsilon}\bv_{\fs,\infty}(\hat{x},\alpha,f_{\fp}(\alpha),f_{\fs}(\alpha)) &= \frac{-\ks^2}{4\pi\omega^2} \sum_{l=1}^M   (I-\hat{x}\hat{x}^T)e^{-\mi\ks\hat{x}\cdot s_l}  \bv^{l}_{\infty}(\alpha,f_{\fp}(\alpha),f_{\fs}(\alpha)) +\mathcal{O}(\kappa\varepsilon), 
	\end{split}
	\end{eqnarray}
	with $$\bv^{l}_{\infty}(\alpha,f_{\fp}(\alpha),f_{\fs}(\alpha)) = \sum_{i=1}^3 \left(e^{\mi\kp\alpha\cdot s_l}f_{\fp}(\alpha)+e^{\mi\ks\alpha\cdot s_l}f_{\fs}(\alpha)\right)\cdot \boldsymbol{e}_i \int_{\Gamma_l} T_{\nu}\bv_{i,l}ds_y. $$
\end{theorem}
Let us define the limit far field operator $F^0$ for elastic scattering of small particles as
\begin{eqnarray}\label{limoper}
	\begin{split}
	F^0(f) = \frac{1}{\omega}&\left(\frac{-\kp^2}{4\pi\omega^2} \sum_{l=1}^M \hat{x}\hat{x}^T  e^{-\mi\kp\hat{x}\cdot s_l}\int_{\ms} \bv^{l}_{\infty}(\alpha,f_{\fp}(\alpha),f_{\fs}(\alpha))ds_{\alpha},\right. \\
	&\left. \frac{-\ks^2}{4\pi\omega^2} \sum_{l=1}^M   (I-\hat{x}\hat{x}^T)e^{-\mi\ks\hat{x}\cdot s_l}\int_{\ms} \bv^{l}_{\infty}(\alpha,f_{\fp}(\alpha),f_{\fs}(\alpha))ds_{\alpha}  \right)
\end{split}
\end{eqnarray} 
and the elastic polarizability tensor for $D_l$, $ 1\le l\le M$, as $$P_l=\left[\int_{\Gamma_l} T_{\nu}\bv_{1,l}ds_y, \int_{\Gamma_l} T_{\nu}\bv_{2,l}ds_y, \int_{\Gamma_l} T_{\nu}\bv_{3,l}ds_y\right].$$ In order to achieve selective focusing on particle $D_l$, we need to investigate the eigenfunciton of the limit far field operator $F^0$, which is stated by the following theorem. 
\begin{theorem}\label{eignapprox}
	 For $1\le l\le M$, assume $P_l$ can be diagonalized by an orthonormal basis $(\boldsymbol{e}_{1,l},\boldsymbol{e}_{2,l},\boldsymbol{e}_{3,l})$ in $\mathbb{R}^3$ with diagonal elements given by $\lambda_{l,i}$,$i=1,2,3$. Define the function $f_{l,i}=(f_{\fp,l,i},f_{\fs,l,i})\in L^2_{\fp}(\ms) \times L^2_{\fs}(\ms)$ to be 
	 \begin{eqnarray}\label{approxeig}
	 	\begin{cases}
	 		f_{\fp,l,i}(\alpha)  = \kp^2(\boldsymbol{e}_{i,l} \cdot \alpha)\alpha e^{-\mi \kp \alpha\cdot s_l}, \\
	 		f_{\fs,l,i}(\alpha)  = \ks^2(\alpha\times(\boldsymbol{e}_{i,l} \times \alpha)) e^{-\mi \ks \alpha\cdot s_l}.
	 	\end{cases}	 	 
	 \end{eqnarray}
 	Then the functions $f_{l,i}(\alpha)$, $1\le l\le M, i=1,2,3,$ are the approximate eigenfunctions of the limit far field operator $F^0$  and satisfy
 	\begin{eqnarray}\label{limitfar}
 		F^0(f_{l,i}) = -\frac{\kp^2+2\ks^2}{3\omega^3}\lambda_{l,i}f_{l,i} + \mathcal{O}((\kappa d)^{-N}), \mbox{ for all } N\in \mathbb{N}, 
 	\end{eqnarray}
 as $d\rightarrow\infty$, where $d$ is the distance defined in \eqref{dist} and $\kappa = \min\{\kp,\ks\}$.
\end{theorem}
\begin{proof}
	 The linear independence of $f_{l,i}$ follows from the fact that $f_{l,i}$ is the far field pattern of elastic point source with polarization $\boldsymbol{e}_{i,l}$ and the far field pattern is unique\cite{DR-book2}. For a given $f_{l,i}$, the Herglotz wave takes the form
	\begin{eqnarray}\label{hegsmall}
		\boldsymbol{u}^{\boldsymbol{i}}_{f_{l,i}}(s_j)&=&\int_{\ms}\kp^2(\boldsymbol{e}_{i,l} \cdot \alpha) \alpha e^{\mi \kp \alpha\cdot (s_j-s_l)}ds_{\alpha}+ \int_{\ms} \ks^2(\alpha\times(\boldsymbol{e}_{i,l} \times \alpha)) e^{\mi \ks \alpha\cdot (s_j-s_l)}ds_{\alpha}\notag \\
		&=&\begin{cases}
			\frac{4\pi(\kp^2+2\ks^2)}{3} \boldsymbol{e}_{i,l}, \quad l=j, \\
			 \mathcal{O}((\kappa d)^{-N}), \quad l\ne j, 
		\end{cases}
	\end{eqnarray} 
where the case for $l=j$ follows from a direct computation and the case for $l\ne j$ follows from the stationary phase theorem. Plugging the Herglotz wave \eqref{hegsmall} into the limit far field operator $F^0$ given by \eqref{limoper} yields the result \eqref{limitfar}.
\end{proof}
By combining the results of Theorems \ref{globalf} and \ref{eignapprox}, we see that if particles are small and distant, selective focusing on an individual particle can be achieved. We summarize the result in the following theorem.  
\begin{theorem}
	For $1\le l\le M$, the approximate eigenfunctions $f_{l,i}(\alpha)$, $i=1,2,3$, given by \eqref{approxeig} will selectively focus on the $l$-th particle.
\end{theorem}
Proof is again based on the stationary phase theorem and we simply omit.

\section{Multiple scattering of elastic particles}
 When particles are not small (compared to the wavelength) and the multiple scattering among particles is not negligible, the result  above based on the asymptotic analysis for small particles is not applicable anymore. One must rely on the numerical algorithm to solve the inverse problem. However, since TRM needs the full aperture far field data from all incident directions on $\ms$, it requires a large amount of forward simulations.   A straightforward numerical method can be obtained through the discretization of a coupled boundary integral system resulted from \eqref{singleintrep}. However, due to the existence of multiple scattering, the discretized linear system suffers from slow convergence when an iterative solver is used. Hence, numerical computation for multiple scattering is highly challenging, especially in three dimensions. To overcome this difficulty and accelerate the computation, we proceed by reviewing the multiple elastic scattering of spheres.

\subsection{Scattering from multiple spheres}
We discuss the elastic scattering of multiple spherical particles under the assumption that all the spheres have the same radius and are non-intersecting with each other. Extension to spheres with different sizes is straightforward. Assume there are $M$ spheres and denote $S_l$ the $l$-th sphere. According to \eqref{inspan}, the incoming field on the sphere $S_l$ can be expressed as 
\begin{eqnarray*}
	\boldsymbol{u^i}(x) &=& \sum_{n=0}^{\infty}\sum_{m=-n}^{n}\left(a^l_{n,m}\nabla\times\nabla\times(x_l u_{n,m}^{\ks}(x_l))/(\mi\ks) +b^l_{n,m}\nabla\times(x_l u_{n,m}^{\ks}(x_l)) +c^l_{n,m}\nabla u_{n,m}^{\kp}(x_l) \right), 
\end{eqnarray*}
where  $x_l$ is the local coordinate of $x$ with respect to the center of $S_l$ and $u^{\kappa}_{n,m}(x_l) = j_n(\kappa r_l)Y_n^m(\theta_l,\phi_l)$ is computed in terms of the spherical coordinates $(r_l,\theta_l,\phi_l)$ of the point $x_l$. The scattered field in the exterior of all the spheres can be represented by a sum of outgoing expansions, one centered at each sphere
\begin{eqnarray*}
	\bv(x) &=& \sum_{l=1}^{M} \sum_{n=0}^{\infty}\sum_{m=-n}^{n}\left(\alpha^l_{n,m}\nabla\times\nabla\times(x_l v_{n,m}^{\ks}(x_l))/(\mi\ks) +\beta^l_{n,m}\nabla\times(x_l v_{n,m}^{\ks}(x_l)) +\gamma^l_{n,m}\nabla v_{n,m}^{\kp}(x_l)\right).
\end{eqnarray*}
Note that $v_{n,m}^{\kappa}(x_l)$ is also evaluated in terms of the spherical coordinate $(r_l,\theta_l,\phi_l)$ of the local coordinate $x_l$.
The coefficients $(\alpha^l_{n,m},\beta^l_{n,m}, \gamma^l_{n,m})$ are all unknowns and  no longer trivial to find as compared to the single sphere case. They are determined by a linear system that imposes the rigid boundary condition \eqref{rigidboundary} on each spherical boundary.  In particular, the incoming field for each sphere has two components, one from the known external incident field $\uinc$, and the other from the field scattered from all the other spheres. This results in a dense linear system involving all of the unknowns. To construct such a linear system, we require the translation operator from the outgoing expansion centered at $S_j$ to the incoming expansion centered at $S_l$\cite{M-06}. More specifically, given the outgoing expansion coefficients $\{\alpha^j_{n,m},\beta^j_{n,m},\gamma^j_{n,m}\}$ from the sphere $S_j$, the corresponding field induced on the sphere $S_l$ can be expanded by the incoming coefficients $\{a^{lj}_{n,m},b^{lj}_{n,m},c^{lj}_{n,m}\}$ with
\begin{eqnarray}
	\begin{cases}
		a^{lj}_{n,m}= \sum^{\infty}_{n'=0}\sum^{n'}_{m'=-n'}\left\{T_{l,j}^{a,\alpha}(n,m,n',m')\alpha_{n',m'}^j + T_{l,j}^{a,\beta}(n,m,n',m')\beta_{n',m'}^j \right\},\\
		b^{lj}_{n,m}= \sum^{\infty}_{n'=0}\sum^{n'}_{m'=-n'}\left\{T_{l,j}^{b,\alpha}(n,m,n',m')\alpha_{n',m'}^j + T_{l,j}^{b,\beta}(n,m,n',m')\beta_{n',m'}^j \right\},\\
		c^{lj}_{n,m}= \sum^{\infty}_{n'=0}\sum^{n'}_{m'=-n'}T_{l,j}^{c, \gamma}(n,m,n',m')\gamma_{n',m'}^j, 
	\end{cases}
\end{eqnarray}
where $T_{l,j}^{a,\alpha},T_{l,j}^{a,\beta}, T_{l,j}^{b,\alpha},T_{l,j}^{b,\beta}$ and $T_{l,j}^{c,\gamma}$ are translation operators. Their explicit expressions can be derived by addition theorem\cite{M-06}, and given in \cite{Xu1996} based on the Gaunt coefficients. However, these expressions are highly involved and suffer from numerical instability.


For the ease of implementation and stability, here we make use of the projection formulas (see \eqref{proj1}-\eqref{proj3}) to obtain the incoming coefficients $\{a^{lj}_{n,m},b^{lj}_{n,m},c^{lj}_{n,m}\}$ by evaluating the incident field on the boundary of $S_l$ induced by the outgoing expansion coefficients $\{\alpha^j_{n,m},\beta^j_{n,m},\gamma^j_{n,m}\}$ from sphere $S_j$. This is numerically much more stable and can be accelerated by the Fast Fourier Transform(FFT). 
More specifically,  for a given incident wave $\uinc=(\uinc_{\fs},\uinc_{\fp})$ on a sphere $S_0$, either from the external plane wave or the field induced from the scattered wave of other particles, one can derive the following result using the orthogonality of spherical harmonics,
\begin{eqnarray}
	a_{n,m} &=& \frac{\mi\ks R}{n(n+1)(j_n(\ks R)+\ks Rj'_n(\ks R) )}\int_{S_0} \uinc_\fs\cdot {\rm Grad}Y_{n}^{-m}(\hat{x})ds, \label{proj1}\\
	b_{n,m} &=&  \frac{-1}{  n(n+1) j_n(\ks R)}\int_{S_0} \uinc_\fs \cdot \hat{x}\times {\rm Grad}Y_{n}^{-m}(\hat{x}) ds, \label{proj2}\\
	c_{n,m} &=& \frac{1}{\kp j'_n(\kp R)}\int_{S_0} \uinc_\fp\cdot Y_n^{-m}(\hat{x})\hat{x}ds, \label{proj3}
\end{eqnarray}
based on the identities
\begin{eqnarray*}
	\begin{split}
		\nabla\times\nabla\times(x u_{n,m}^{\ks}) & = \nabla(u_{n,m}^{\ks}+x\cdot \nabla u_{n,m}^{\ks}) + \ks^2 x u_{n,m}^{\ks},\\ 
		\nabla\times(x u_{n,m}^{\ks}) & = \nabla u_{n,m}^{\ks}\times x,\\
		\nabla u_{n,m}^{\kp} &= \kp j'_n(\kp R) Y_n^m(\hat{x}) + \frac{j_n(\kp R)}{R}{\rm Grad} Y_n^m(\hat{x}).
	\end{split}
\end{eqnarray*}
\begin{remark}
	The denominator in \eqref{proj1}-\eqref{proj3} might be zero for some resonance value of $j_n(r)$. However, the denominator in \eqref{proj1} and \eqref{proj2} can not be simultaneously zero, since all the zeros of $j_n(r)$ are simple\cite{DR-book2}. Therefore, a remedy is to use $$\nabla\times\uinc_\fs/(\mi\ks)=\sum_{n=0}^{\infty}\sum_{m=-n}^{n}\left(b_{n,m}\nabla\times\nabla\times(x u_{n,m}^{\ks})/(\mi\ks) -a_{n,m}\nabla\times(x u_{n,m}^{\ks})\right)$$ on $S_0$ to find $a_{n,m}$ or $b_{n,m}$ when resonance happens. Similarly, one can use 
	\[
	c_{n,m} = \frac{R}{n(n+1)j_n(\kp R)}\int_{S_0} \uinc_\fp\cdot {\rm Grad}Y_n^{-m}(\hat{x})\hat{x}ds.
	\]
	to find $c_{n,m}$ when $j'_n(\kp R)=0$. For numerical stability, the two ways can be combined together by least square. We skip the details by simply assuming that $\kp$ and $\ks$ are not the resonance frequencies in $S_0$ for any $n\ge 0$.
\end{remark}

To sum up, the total elastic field on the sphere $S_l$ is given by
\begin{eqnarray}
	\boldsymbol{u}&=&\uinc + \sum_{j\ne l}\sum_{n=0}^{\infty}\sum_{m=-n}^{n}\left(a^{lj}_{n,m}\nabla\times\nabla\times(x_l u_{n,m}^{\ks}(x_l))/(\mi\ks) +b^{lj}_{n,m}\nabla\times(x_l u_{n,m}^{\ks}(x_l)) +c^{lj}_{n,m}\nabla u_{n,m}^{\kp}(x_l)\right) \notag \\ &&+\sum_{n=0}^{\infty}\sum_{m=-n}^{n}\left(\alpha^l_{n,m}\nabla\times\nabla\times(x_l v_{n,m}^{\ks}(x_l))/(\mi\ks) +\beta^l_{n,m}\nabla\times(x_l v_{n,m}^{\ks}(x_l)) + \gamma^l_{n,m}\nabla v_{n,m}^{\kp}(x_l)\right).
\end{eqnarray}
It can been seen that the first term in the preceding expressions accounts for the external incoming field, while the next term accounts for the scattered field coming from all other spheres. The last term accounts for the fields being scattered by $S_l$ itself. Based on the rigid boundary condition \eqref{rigidboundary}, it ends up to be a linear system 
\begin{eqnarray}\label{multiequ}
	\mathcal{S}^{-1}\begin{bmatrix}
		\alpha^l_{n,m} \\
		\beta^l_{n,m} \\
		\gamma^l_{n,m} \\
	\end{bmatrix}-\sum_{j\ne l}\begin{bmatrix}
		T_{l,j}^{a,\alpha}& T_{l,j}^{a,\beta} & 0  \\
		T_{l,j}^{b,\alpha}& T_{l,j}^{b,\beta} & 0 \\
		0& 0 &  T_{l,j}^{c,\gamma}\\
	\end{bmatrix}
	\begin{bmatrix}
		\alpha^j_{n,m} \\
		\beta^j_{n,m} \\
		\gamma^j_{n,m} \\
	\end{bmatrix}
	=\begin{bmatrix}
		a^l_{n,m} \\
		b^l_{n,m} \\
		c^l_{n,m} \\
	\end{bmatrix},
\end{eqnarray} 
with $ n= 0,1,\cdots,-n\le m\le n, l = 1,2,\cdots, M.$ Equation \eqref{multiequ} is the multiple scattering system for elastic spheres. 

\subsection{Scattering of multiple general-shaped particles}
We now extend the multiple scattering theory of spherical particles to the non-spherical case. For simplicity, we assume particles are well-separated in the sense that each particle is enclosed in a sphere and all the spheres are non-intersecting. In addition, assume all the particles are the same up to a rotation. In other words, every particle $D_l$ is obtained by a translation and rotation from a reference particle $D_0$ centered at the origin,
\[D_l = \{x=s_l+\mathcal{Q}_ly; y\in D_0\},\]
where $s_l\in\mathbb{R}^3$ and $\mathcal{Q}_l$ is a $3\times3$ rotation matrix. The main difference from the previous case to the current case is that we need a scattering matrix $\mathcal{S}$ for the non-spherical obstacle $D_0$, which maps the incoming field to the scattered field of $D_0$. This can not be obtained analytically but can be constructed numerically through the boundary integral equation. More specifically, denote $\Gamma_0$ the boundary of $D_0$.  Given an incident wave $\uinc$, we can represent the scattered wave $\bv$ by the single layer potential \eqref{singleintrep} (with $\Gamma$ replaced by $\Gamma_0$) and solve the boundary integral equation
\begin{eqnarray}\label{bdintegral}
	\int_{\Gamma_0}\Phi(x,y)\bJ(y)ds_y = -\uinc, \quad x\in \Gamma_0.
\end{eqnarray}
Note that the boundary integral equation \eqref{bdintegral} is only uniquely solvable when $\ks$ or $\kp$ is not the interior Dirichlet eigenvalue of the elastic equation \eqref{navierequ}. We assume this is satisfied for $D_0$ since the integral equation based on combined layer potential\cite{Louer2014} can be used if it is not the case. 
Suppose now that we have the obstacle $D_0$ enclosed by $S_0$, where $S_0$ is centered at the origin with radius $R$. To numerically construct the scattering matrix on $S_0$, we sequentially choose
\[
 \nabla\times\nabla\times(xu_{n,m}^{\ks})/(\mi \ks),\quad \nabla\times(xu_{n,m}^{\ks}),\quad \nabla u_{n,m}^{\kp},\quad m=-n,\cdots,n,\quad n=0,1,2,\cdots, 
 \]
 as the incident wave $\uinc$ for $D_0$ and solve the boundary integral equation \eqref{bdintegral}. Then evaluate the scattered wave $\bv$ on $S_0$ through \eqref{singleintrep} and convert it into \textit{outgoing expansion coefficients}. 	The formula for converting the scattered wave $\bv$ to the outgoing coefficients on $S_0$ can also be obtained through projection 
\begin{eqnarray}
	\alpha_{n,m} &=& \frac{\mi \ks R}{n(n+1)(h^{(1)}_n(\ks R)+\ks Rh^{(1)'}_n(\ks R) )}\int_{S_0} \bv_\fs\cdot  {\rm Grad}Y_{n}^{-m}(\hat{x})ds, \label{proj4}\\
	\beta_{n,m} &=&  \frac{-1}{  n(n+1)\mi \ks h^{(1)}_n(\ks R)}\int_{S_0} \bv_\fs \cdot \hat{x}\times {\rm Grad}Y_{n}^{-m}(\hat{x}) ds, \label{proj5}\\
	\gamma_{n,m} &=& \frac{1}{\kp h^{(1)'}_n(\kp R)}\int_{S_0} \bv_\fp\cdot Y_n^{-m}(\hat{x})\hat{x}ds. \label{proj6}
\end{eqnarray}
which is the same as \eqref{proj1}-\eqref{proj3} by simply replacing $j_n(r)$ by $h_n^{(1)}(r)$. 
Once the scattering matrix $\mathcal{S}$ is known, the solution to the full elastic equation for $M$ general-shaped particles can be turned into a multiple scattering problem of the enclosing spheres, which is exactly the same as equation \eqref{multiequ}.

\subsection{Fast algorithm for multiple particles scattering}
Let us now assume all the incoming and outgoing expansion are truncated at $n=N$ terms.   When the number of particles $M$ is large and all the particles have non-negligible interaction, direct solving \eqref{multiequ} will be difficult since the system is large and dense. In addition, such a system is very ill-conditioned since the elements in the scattering matrix $\mathcal{S}$ is exponentially decreasing as $n$ increases, as can be seen for the single sphere case. A better formulation is to multiply both sides by $\mathcal{S}$, which yields
\begin{eqnarray}\label{multiequ2}
	\left(I-\mathcal{S}\sum_{j\ne l}\begin{bmatrix}
		T_{l,j}^{a,\alpha}& T_{l,j}^{a,\beta} & 0  \\
		T_{l,j}^{b,\alpha}& T_{l,j}^{b,\beta} & 0 \\
		0& 0 &  T_{l,j}^{c,\gamma}\\
	\end{bmatrix}\right)
	\begin{bmatrix}
		\alpha^j_{n,m} \\
		\beta^j_{n,m} \\
		\gamma^j_{n,m} \\
	\end{bmatrix}=\mathcal{S}\begin{bmatrix}
		a^l_{n,m} \\
		b^l_{n,m} \\
		c^l_{n,m} \\
	\end{bmatrix}, l = 1,2,\cdots, M,
\end{eqnarray} 
where $I$ is the identity matrix. This is equivalent to applying a block diagonal preconditioner to the system \eqref{multiequ}.  Since the translation matrices $T_{l,j}$ do not involve self interactions, the preconditioned system \eqref{multiequ2} is a discrete analogue of a second kind boundary integral equation, which is much better conditioned than the system \eqref{multiequ} and converges rapidly when iterative solver like GMRES is used. Since the matrix is dense, straightforward implementation of the iterative solver suffers from expensive matrix–vector multiply and requires $\mathcal{O}(M^2N^3)$ work in each iterative process. In order to accelerate the solution procedure, we apply the fast multipole method (FMM) to speed up the matrix–vector product, which can easily reduce the cost to $\mathcal{O}(MN^3)$ work per iteration.  As the literature on FMMs is substantial, we omit the technical details and refer interested readers to \cite{R-90,GG2013} for detailed information. 

To sum up, the multiple scattering method based on expansion coefficients not only reduces the number of degrees of freedom, but also pre-computes the solution for each particle in isolation, so that the linear system we solve by iteration on the multi-sphere system is much better conditioned and can be solved by the combination of GMRES and FMM with a cost proportional to $M$. Similar methods have been successfully applied to scattering of multiple particles in the acoustics \cite{LKB-JCP15} and electromagnetics\cite{GG2013}. 
\begin{remark}
	The error analysis of the proposed method is highly involved due to the multiple scattering. Numerically, spectral accuracy can be achieved when the truncation number $N$ increases. However, the error not only depends on the truncation number, but also the wavenumber $\ks$ and $\kp$, the closest distance of two particles, the accuracy for solving the boundary integral equation \eqref{bdintegral} and the FMM. Although there exist solvers for \eqref{bdintegral} with spectral accuracy\cite{Louer2014,JH2022}, and errors for the FMM has been analyzed in \cite{ED20}, detailed analysis for the error of multiple scattering is beyond the scope of this paper. In the case of a single sphere, since all the modes up to $n=N$ is exact, the error is only due to the boundary values of the incident field for modes $n>N$. One can follow the calculation in \cite{GD07} to show that the truncation error for plane wave incidence is  
	\begin{eqnarray}
		|\varepsilon_N| \le C\frac{ e^{\kappa R/2}}{N!}\left(\frac{\kappa R}{2}\right)^N, \mbox{ with } \kappa=\max\{\ks,\kp\},
	\end{eqnarray}
	with $C$ independent of $N$ and $\kappa R$, which is exponentially small for large $N$.
\end{remark}

\section{Numerical experiments}
In this section, we test the algorithm in several examples. The general procedure is: first apply the fast multiple elastic scattering algorithm to obtain the far field data, and then use the time reversal method to recover the location and shape of particles. In the forward simulation, we apply GMRES to solve \eqref{multiequ2} with a residual accuracy  $10^{-6}$ and modify the code of the fast multiple method provided in \cite{GG2013} to accelerate the matrix-vector product. For the inversion part, to avoid inverse crime and test the robustness of the algorithm, we add $5\%$ Gaussian noise to the simulated far field data. Once the numerical far field operator $F$ is found, the eigenvalues and eigenvectors are obtained through the `\textit{eig}' command in \textit{MATLAB}, and the Herglotz wave $\uincf$ is evaluated in a straightforward manner, although NUFFT\cite{GL04} can be used to accelerate the evaluation if time becomes an issue. According to Theorem \ref{globalf}, both the compressional and shear incident wave can be used to focus the particles, but here we only use one component of the shear wave for simplicity.

Throughout all the examples,  the compressional and shear wavenumbers are chosen to be $\kp=\frac{\pi}{3}$ and $\ks=\frac{5\pi}{8}$, and the truncation number for each particle in the expansion \eqref{inspan} is $N=10$. The far field is measured at $231$ different directions in $[0,\pi]\times[0,2\pi] $, with $11$ points in $[0,\pi]$ chosen by Gauss-Legendre nodes and $21$ equally spaced points in $[0,2\pi]$ . 

\subsection{Example 1}
In this example, we try to recover a single elastic sphere centered at $(5,0,0)$ with radius $0.5$. Results are shown in Figure \ref{ex2}, where \ref{ex2}(a) is the exact shape of the sphere. The reconstruction of the sphere is given in \ref{ex2}(b). The cross section in the $z=0$ plane of the magnitude of Herglotz wave generated by the eigenfunction associated with the largest eigenvalue is shown in \ref{ex2}(c) and the full three dimensional plot is given in \ref{ex2}(d).  Clearly, the Herglotz wave function achieves its maximum at the location of the sphere, which shows the time reversal method is very effective in reconstructing the location of a single sphere. To reconstruct the shape, we have to choose an appropriate cut-off value. Such a value is generally empirical, and here we choose the value to be 1 in order to obtain \ref{ex2}(b). 
\begin{figure}
	\centering
	\subfloat[]{}\includegraphics[width=6cm,height=4.8cm]{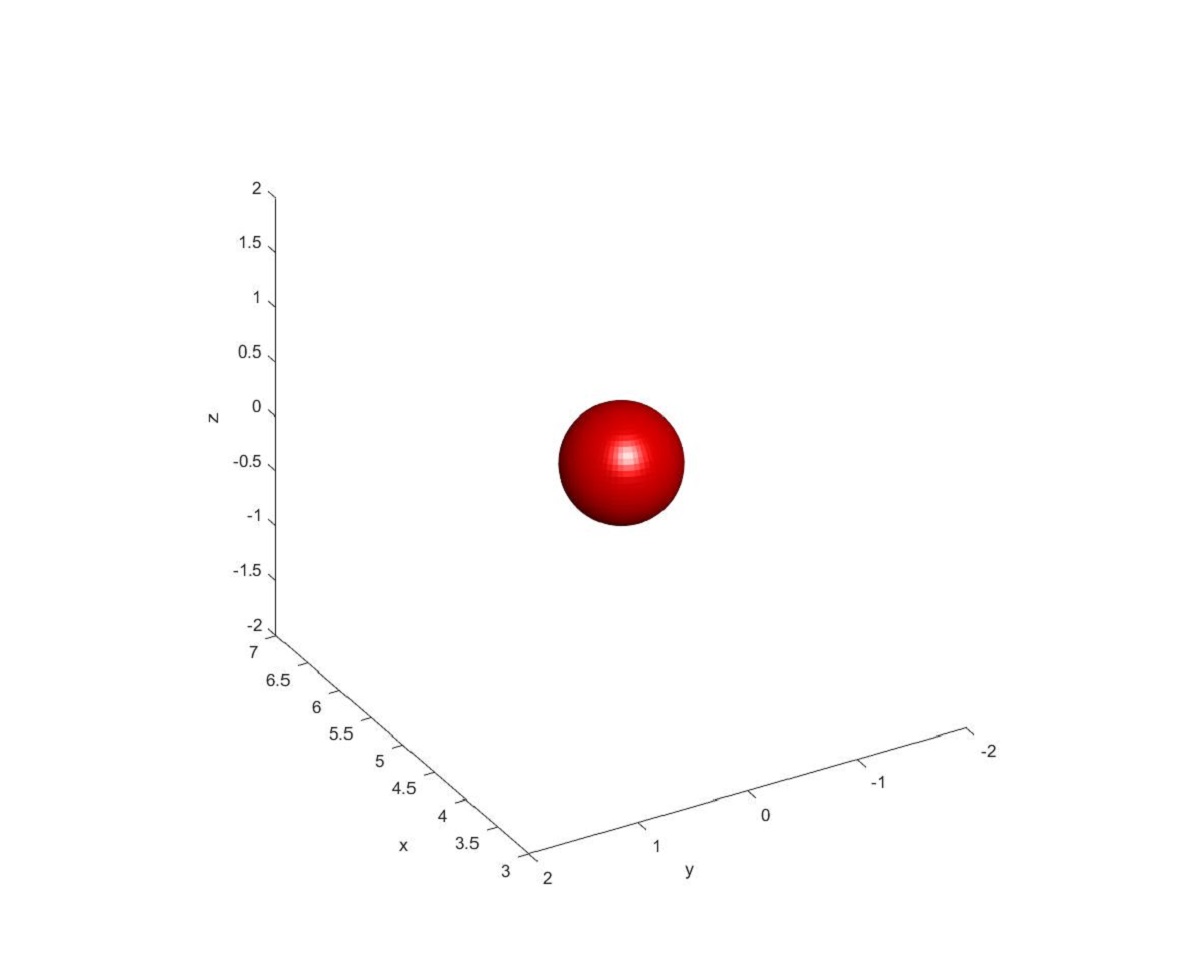}
	\subfloat[]{}\includegraphics[width=6cm,height=4.8cm]{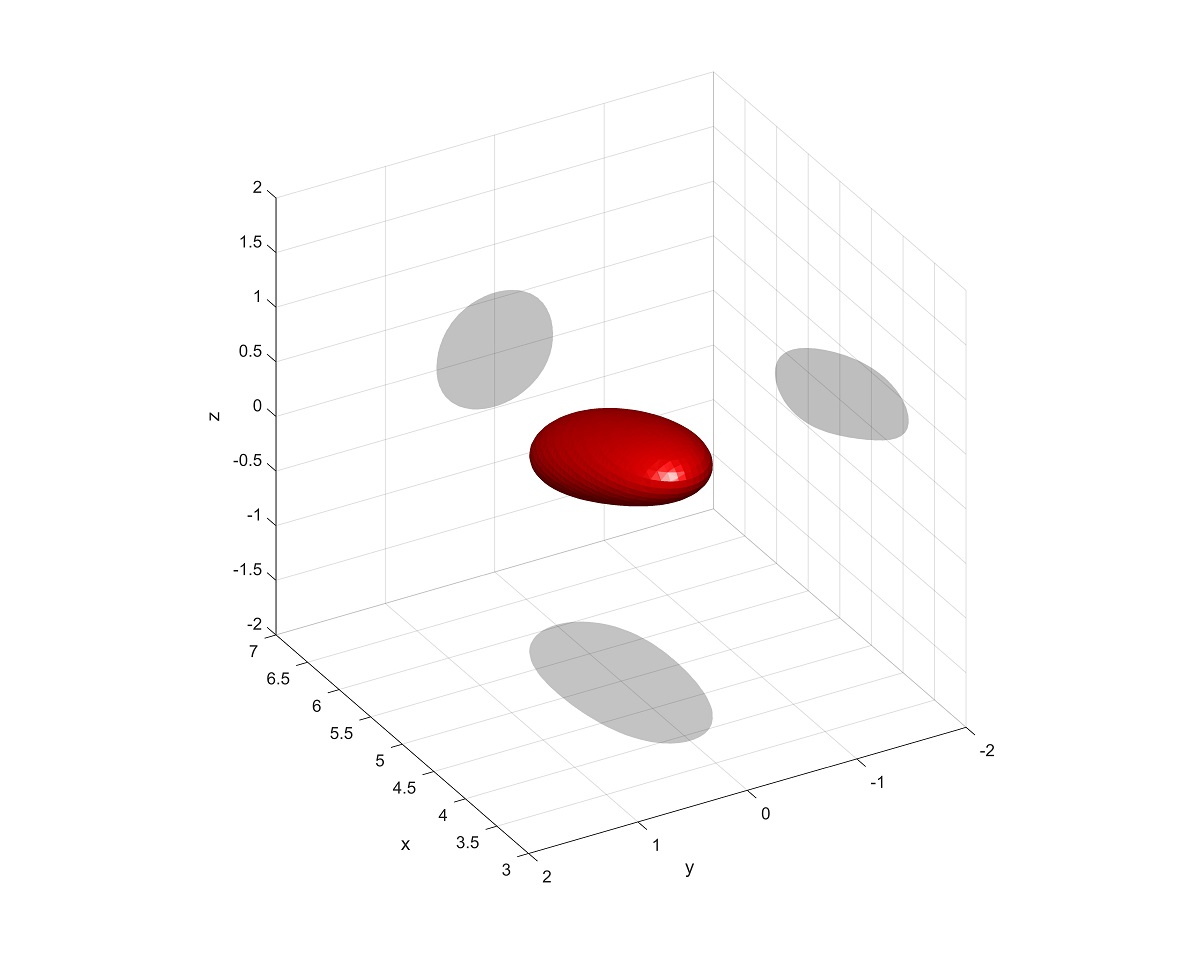}
	\subfloat[]{}\includegraphics[width=6cm,height=4.8cm]{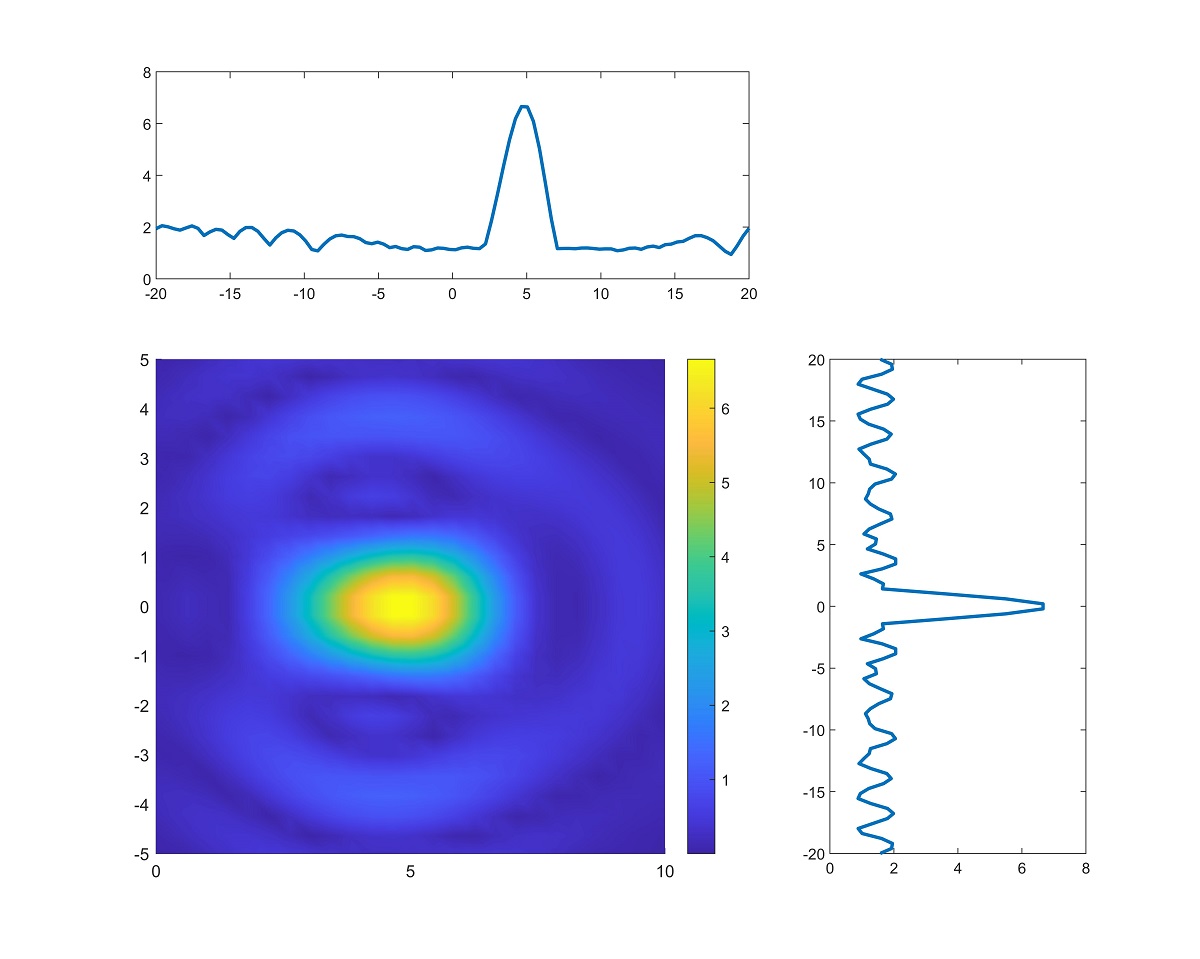}
	\subfloat[]{}\includegraphics[width=6cm,height=4.8cm]{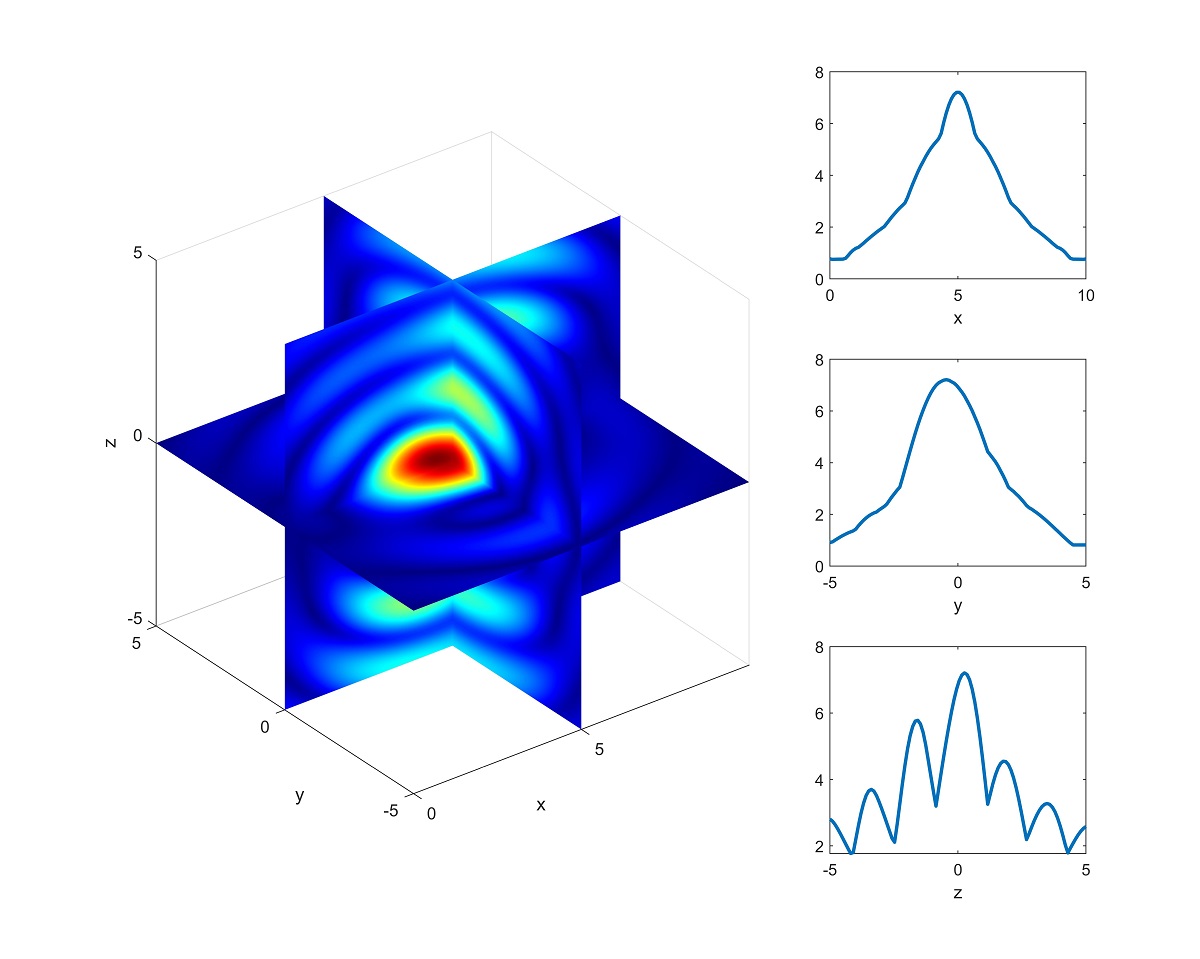}
	\caption{Imaging of a single sphere: (a) Exact shape; (b) Reconstruction based on the TRM; (c) Imaging function in the $z=0$ plane; (d) Imaging function in the $x=0$, $y=0$ and $z=0$ plane.}\label{ex2}
\end{figure}
\subsection{Example 2}
In this example, we test the reconstruction of a rotated \textit{starfish} centered at $(5,0,0)$, as shown in Figure \ref{ex3}(a). The shape is generated by rotating with respect to the $z$ axis with a generating curve given by 
\begin{align*}
	\begin{cases}
	r(t)=[2+0.5\cos(5\pi(t-1))]\cos(\pi(t-0.5)),\\
	z(t)=[2+0.5\cos(5\pi(t-1))]\sin(\pi(t-0.5)).
	\end{cases}
\end{align*}
The reconstructed  shape is show in \ref{ex3}(b), which is obtained by choosing a cut-off value to be 0.8 in the Hergoltz wave function. Again, the Hergoltz wave function is obtained by using the eigenfunction associated with the largest eigenvalue. Figure \ref{ex3}(c) shows that the location can be accurately reconstructed. In \ref{ex3}(d), we are trying to recover the shape of the unknown particle by using a different cut-off value 1.5. It shows the time reversal method can only vaguely reconstruct the shape of the obstacle. In order to accurately reconstruct the shape, an iterative type imaging method may be needed with the combination of the current result. Details will be explored in the future.
\begin{figure}
	\subfloat[]{}\includegraphics[width=6cm,height=4.8cm]{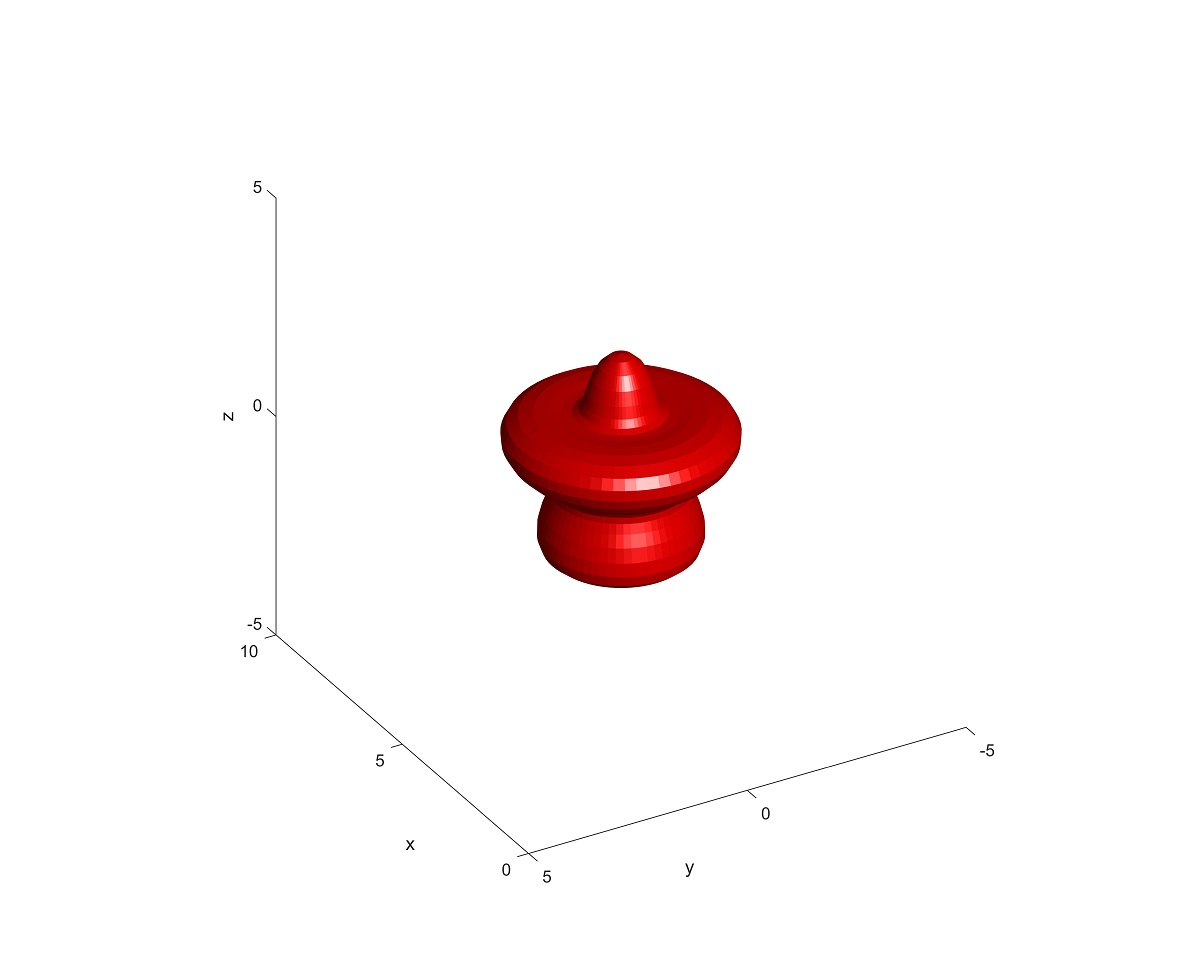}
	\subfloat[]{}\includegraphics[width=6cm,height=4.8cm]{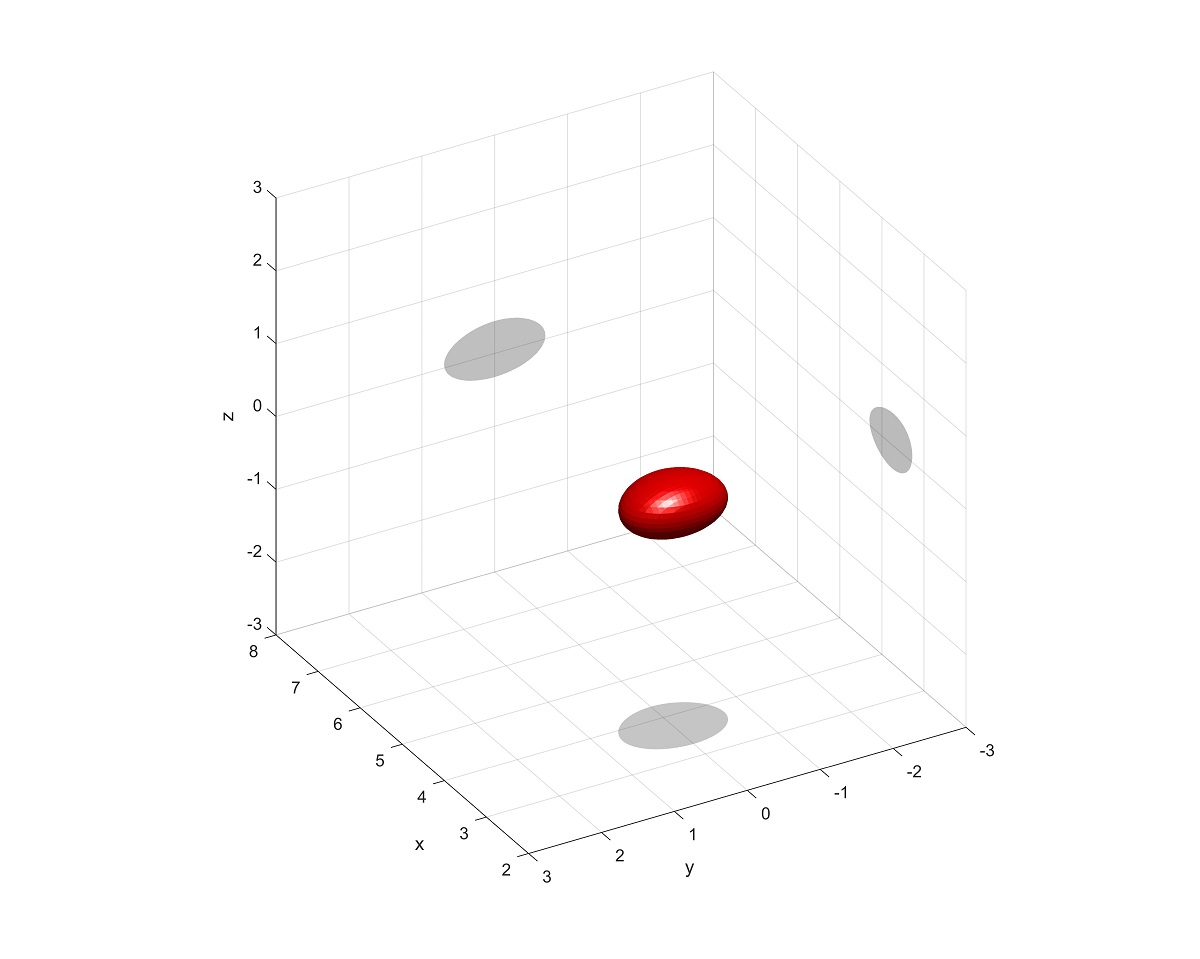}
	\subfloat[]{}\includegraphics[width=6cm,height=4.8cm]{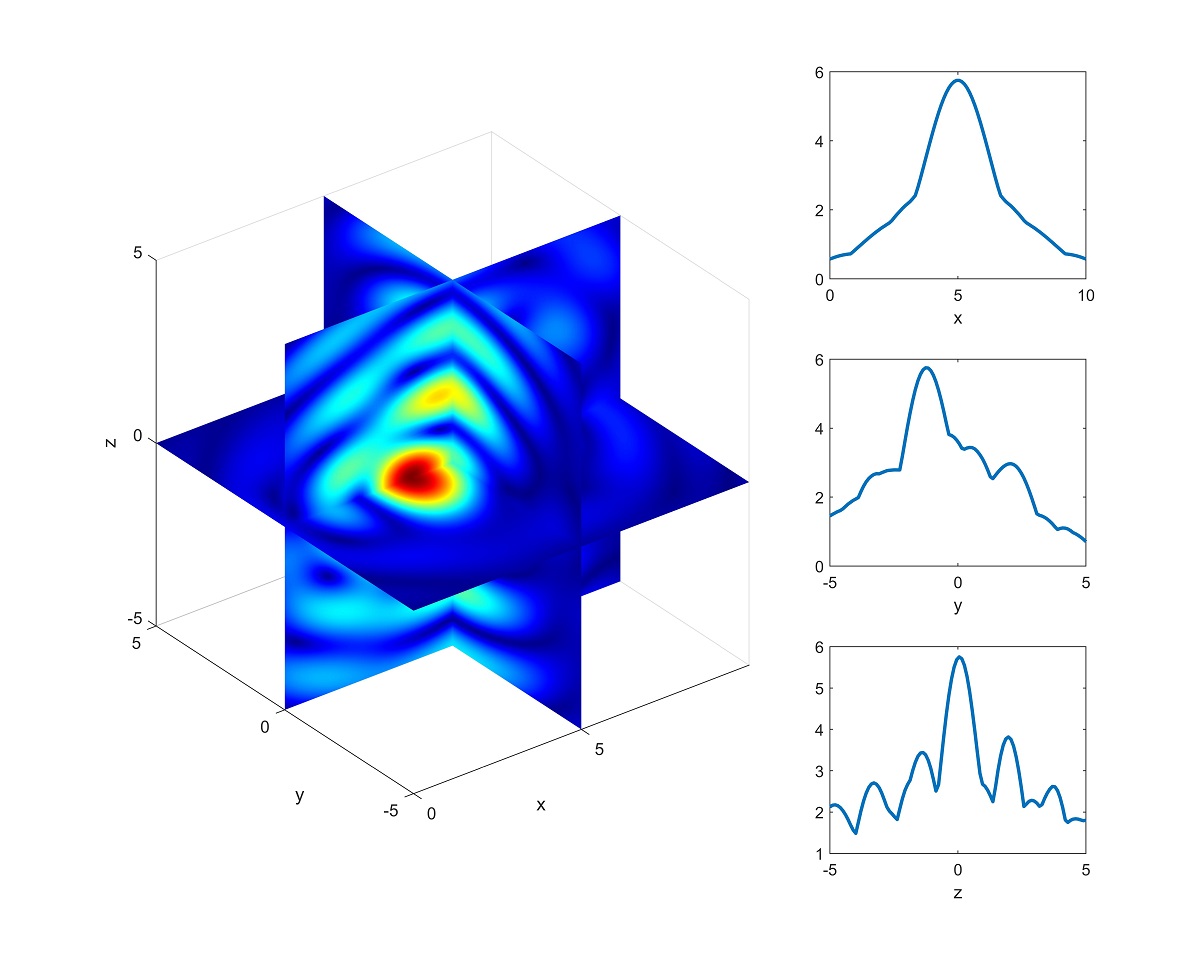}
	\subfloat[]{}\includegraphics[width=6cm,height=4.8cm]{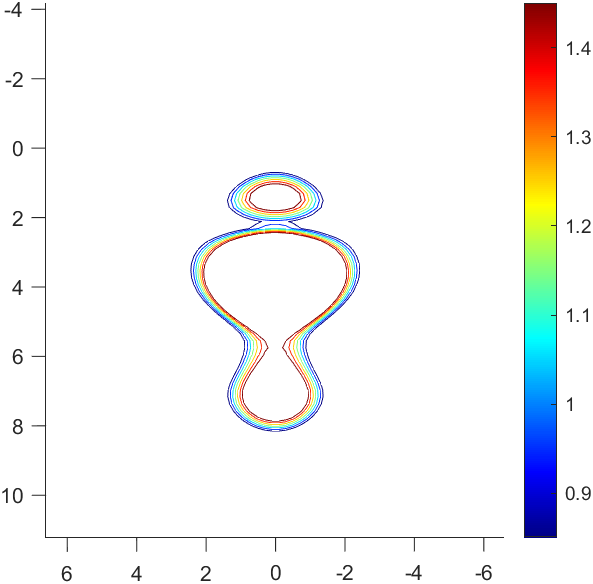}
	\caption{Imaging of a rotated starfish: (a) Exact shape; (b) Reconstruction based on the TRM with a cut-off value $0.8$; (c) Imaging function in the $x=0$, $y=0$ and $z=0$ plane;(d) Imaging the shape through a cut-off value being $1.5$.}\label{ex3}
\end{figure}
\subsection{Example 3}
In this example, we are trying to recover the locations of three spheres simultaneously. The exact locations of the three spheres are shown in Figure \ref{ex4}(a). The recovered obstacles are given in \eqref{ex4}(b). In order to do the reconstruction, we use  the Herglotz wave generated by the eigenfunction associated with the largest eigenvalue. Figure \eqref{ex4}(c) gives the magnitude of the Herglotz wave in the $z=0$ plane and  \eqref{ex4}(d) plots the Herglotz wave function in all the three planes. It clearly shows the location of the spheres can be globally illuminated by the Herglotz function.
\begin{figure}
	\subfloat[]{}\includegraphics[width=6cm,height=4.8cm]{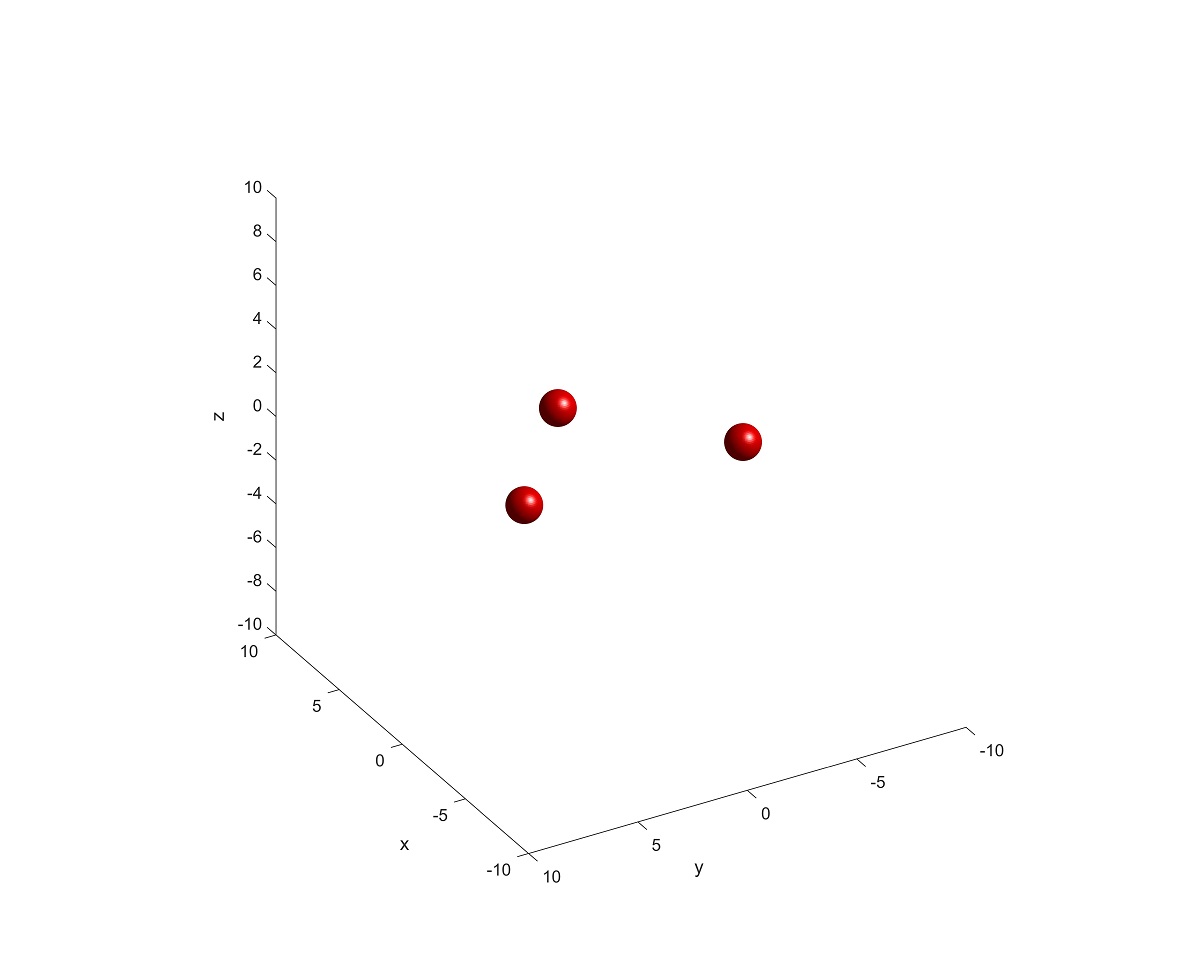}
	\subfloat[]{}\includegraphics[width=6cm,height=4.8cm]{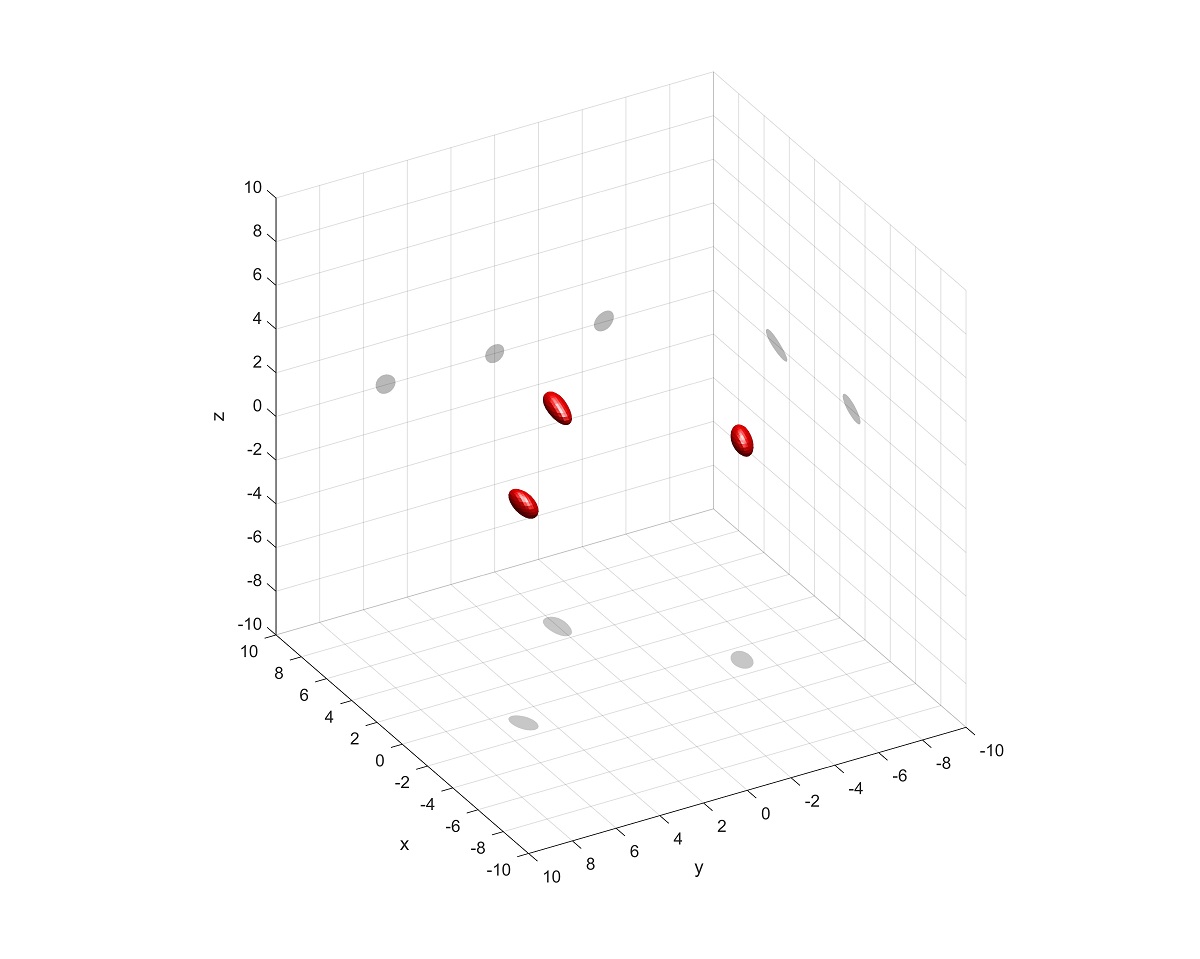}
	\subfloat[]{}\includegraphics[width=6cm,height=4.8cm]{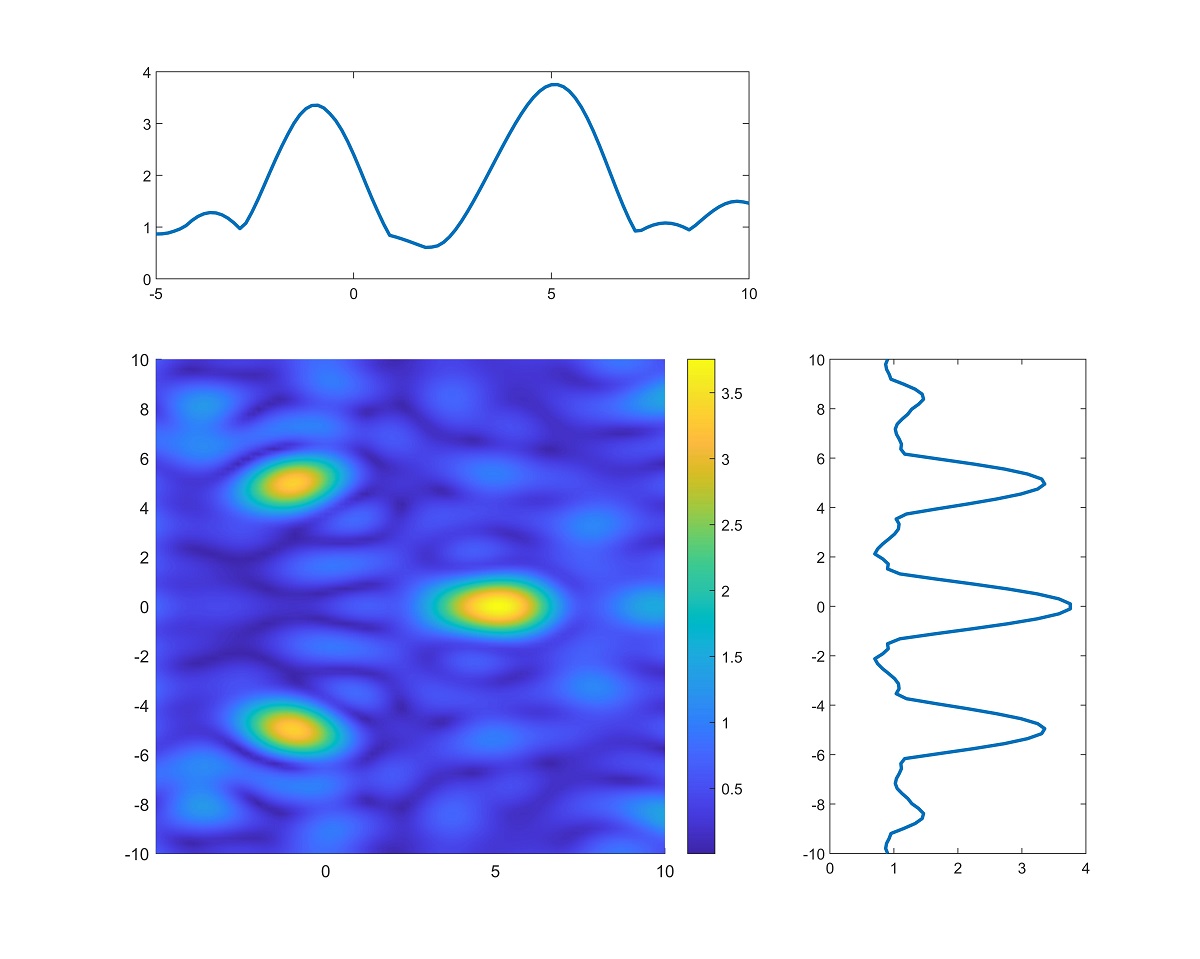}
	\subfloat[]{}\includegraphics[width=6cm,height=4.8cm]{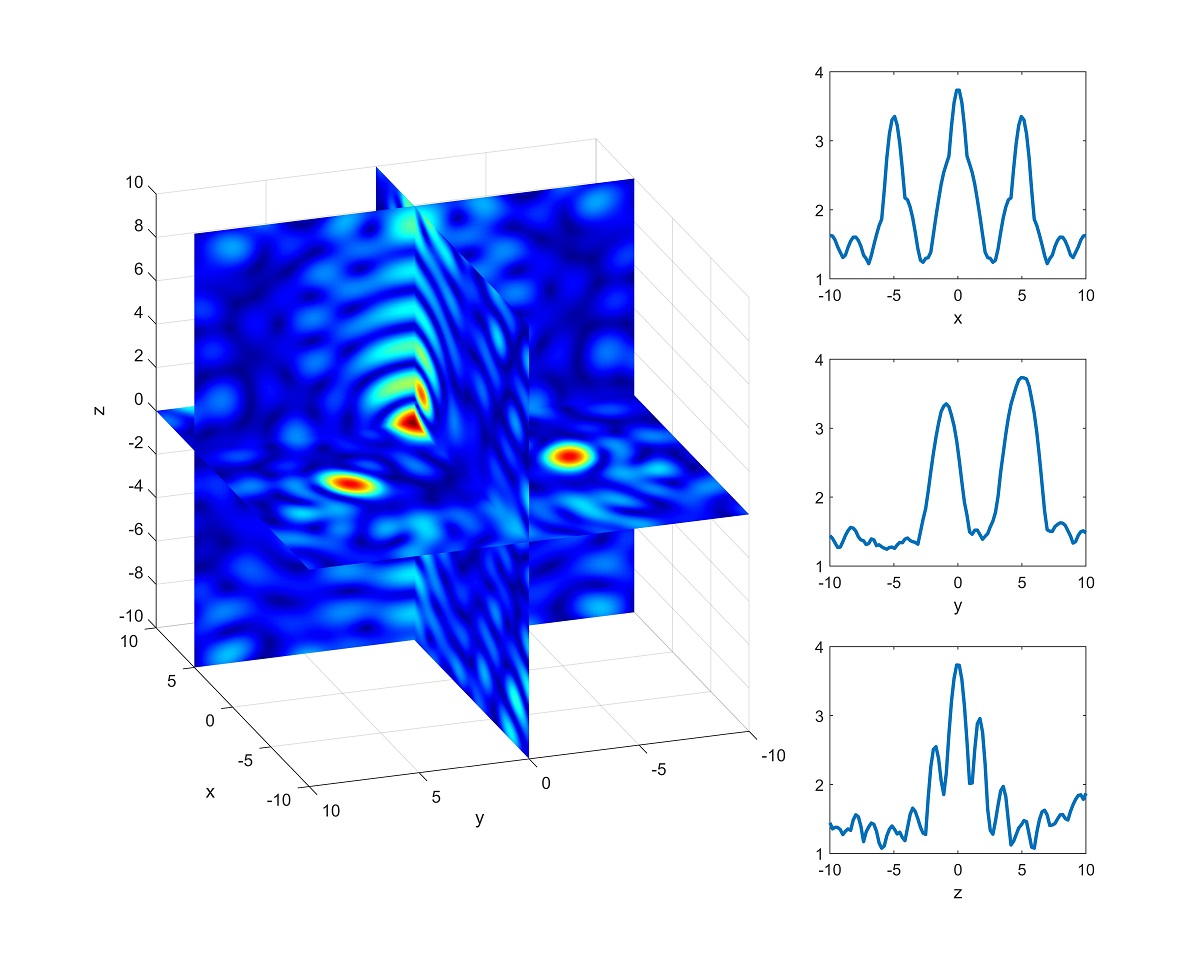}
	\caption{Imaging of three spheres: (a) Exact shape; (b) Reconstruction based on the TRM; (c) Imaging function in the $z=0$ plane; (d) Imaging function in the $x=0$, $y=0$ and $z=0$ plane.}\label{ex4}
\end{figure}
\subsection{Example 4}
In this example, we test the algorithm to recover multiple elastic particles. In particular, Figure \eqref{ex5}(a) gives the exact location of five spheres and \eqref{ex5}(b) shows the recovered obstacles based on TRM. Figure \eqref{ex5}(c) gives the exact location of 27 ellipsoids and \eqref{ex5}(d) shows the recovered obstacles based on TRM. In both cases, the locations of all the particles are accurately reconstructed. It is worth mentioning that our fast algorithm is able to complete the simulation in less than 1000 seconds. While most of the time is spent on the forward simulation, the time for the inversion part is negligible. This demonstrates the effectiveness of the proposed method. 
\begin{figure}
	\subfloat[]{}\includegraphics[width=6cm,height=4.8cm]{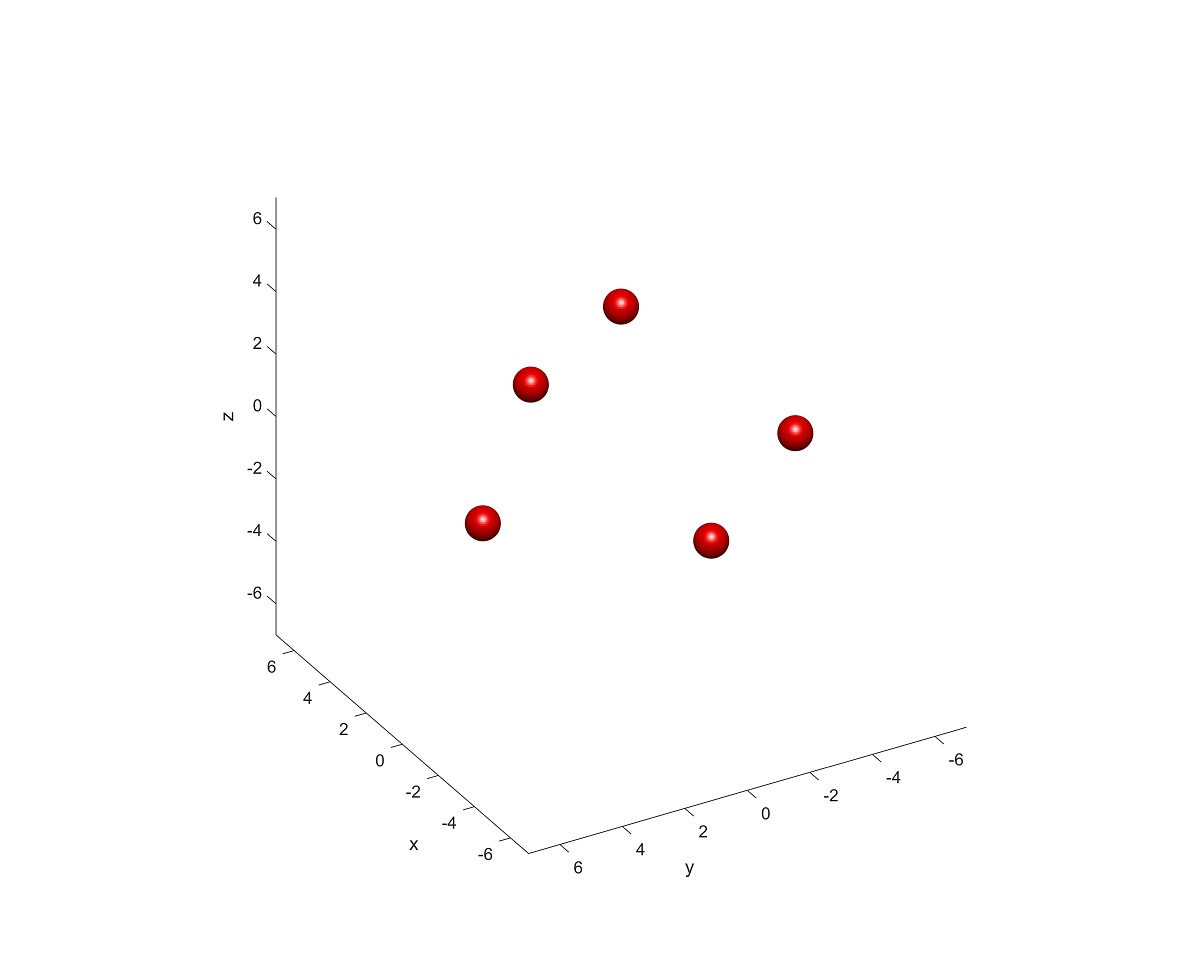}
	\subfloat[]{}\includegraphics[width=6cm,height=4.8cm]{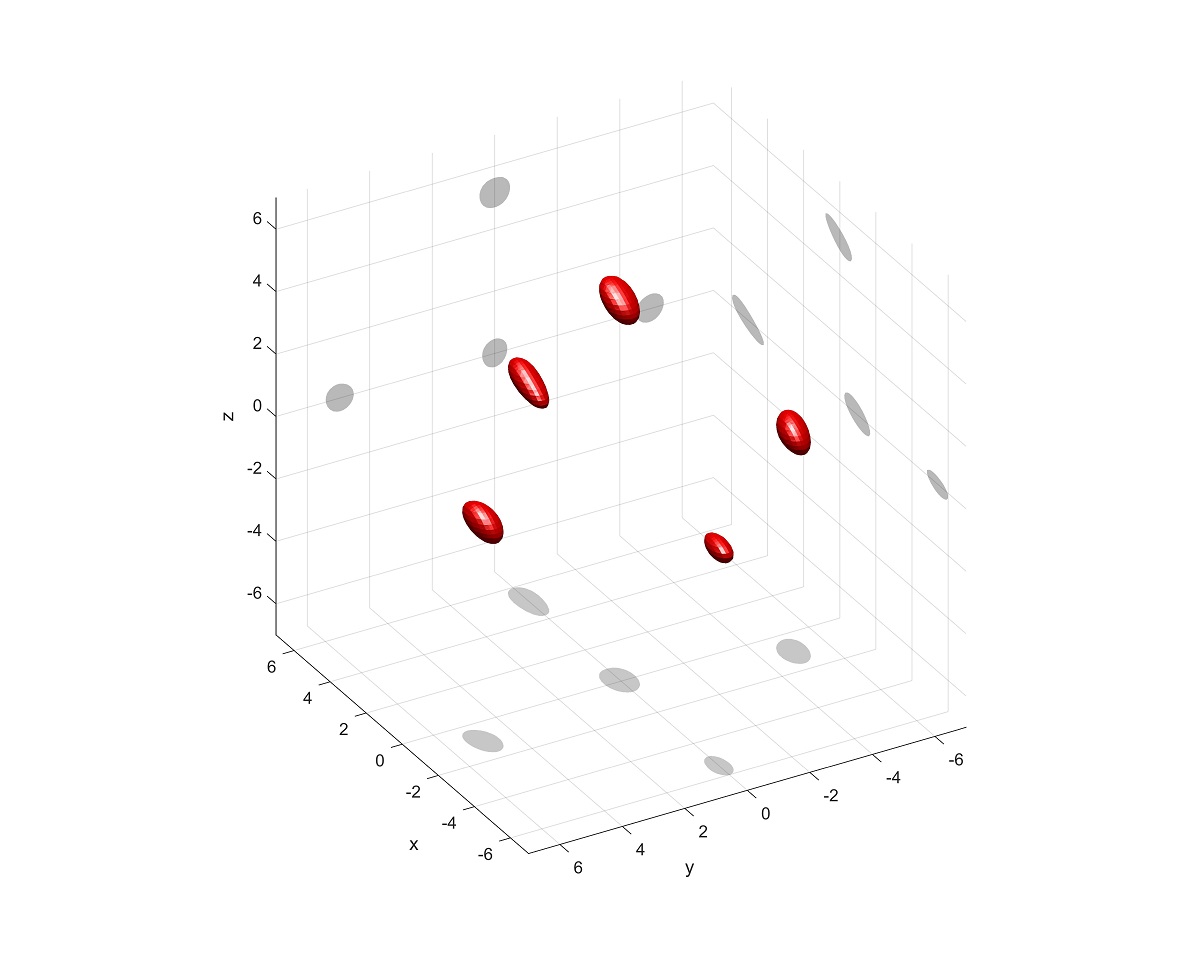}
	\subfloat[]{}\includegraphics[width=6cm,height=4.8cm]{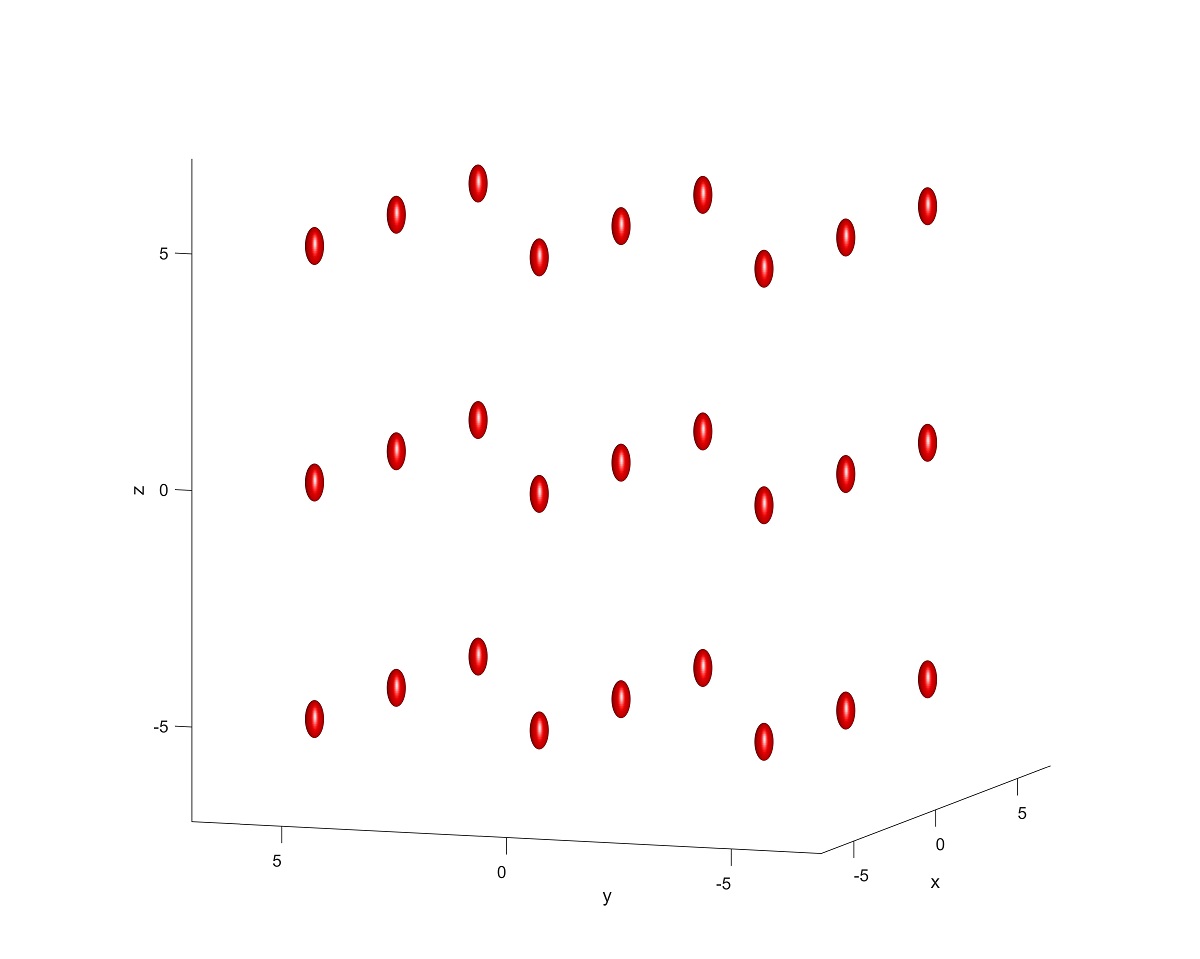}
	\subfloat[]{}\includegraphics[width=6cm,height=4.8cm]{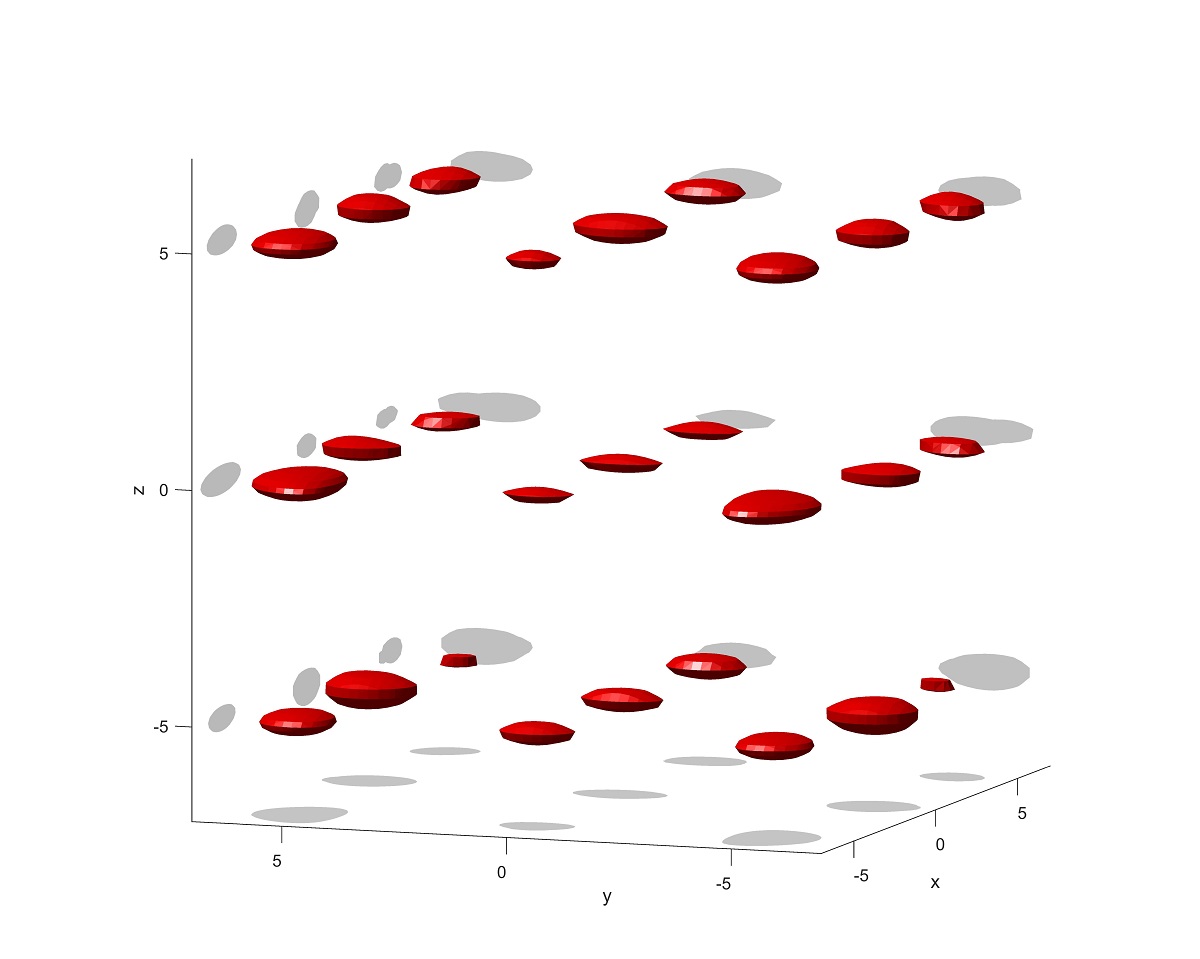}
	\caption{Imaging of multiple particles: (a) Exact shape of $5$ spheres; (b) Reconstruction based on the TRM; (c) Exact shape of $27$ ellipsoids; (d) Reconstruction based on the TRM.}\label{ex5}
\end{figure}
\section{Conclusion}
In this paper,  we analyze the time reversal method for the inverse elastic scattering of multiple particles. We show that selective focusing can be achieved for small and distant particles. Detailed expression for the approximate eigenvectors is given based on the asymptotic analysis of independent elastic scattering. For particles with non-negligible multiple scattering,  we propose a fast algorithm for evaluating the scattering of  multiple elastic particles by combining multiple scattering theory and fast multiple method. The method greatly reduces the number of unknowns and vastly decrease the iteration number in GMRES based on the scattering matrix preconditioner.  Numerical experiments show that the combination of TRM and the fast solver is very effective to simulate the inversion of multiple particles in the existence of non-negligible multiple scattering. Our future work includes developing high resolution imaging method for multiple elastic particles and extending the inversion method to recover particles with other kinds of elastic boundary conditions. 
\appendix
\section{Plane wave expansion}\label{appA}
For an elastic plane wave given by: 
	$$\mathbf{u}^{plane}(\bx) = \frac{1}{\mu} e^{i\ks \bx \cdot \bd}(\bd\times \bp)\times \bd+\frac{1}{\lambda+2\mu}e^{i\kp \bx \cdot \bd}(\bd\cdot \bp) \bd,$$
	where $\bp$ is the polarization direction and $\bd$ is the propagation direction, the incoming expansion is given by
	\begin{eqnarray}\label{planespan}
		\begin{split}
			a_{n,m} &= \frac{1}{\mu}\frac{4\pi \mi^n}{n(n+1)}{\rm Grad}Y_n^{-m}(\bd)\cdot \bp \\
			b_{n,m} &= \frac{1}{\mu}\frac{4\pi \mi^n}{n(n+1)}\bd\times{\rm  Grad}Y_n^{-m}(\bd)\cdot \bp  \\
			c_{n,m} &= -\frac{4\pi \mi^{n+1}}{\kp (\lambda+2\mu)}Y_n^{-m}(\bd)\bd \cdot \bp 
		\end{split}
	\end{eqnarray}
\section{Scattering matrix of a sphere}\label{appB}
	According to the boundary condition \eqref{rigidboundary} on $S_0$,  mode matching yields 
	\begin{eqnarray}\label{scat_spher}
		\begin{bmatrix}
			\alpha_{n,m} \\
			\beta_{n,m} \\
			\gamma_{n,m} \\
		\end{bmatrix}
		=\mathcal{S}_{n,m}
		\begin{bmatrix}
			a_{n,m} \\
			b_{n,m} \\
			c_{n,m} \\
		\end{bmatrix},
		 m=-n,\cdots, n, \quad n = 0,1,2,\cdots, 
	\end{eqnarray} 
Here we let $\alpha_{0,0}=\beta_{0,0}=a_{0,0} = b_{0,0}=0$ and $\mathcal{S}_{n,m}$ be a $3\times 3$ block. It holds from \cite{Louer2014} that $\gamma_{0,0} = -\frac{j'_0(\kp R)}{{h_0^{(1)}}'(\kp R)}c_{0,0}$, and $\beta_{n,m}=-\frac{j_n(\ks R)}{{h_n^{(1)}}(\ks R)}b_{n,m}$ for $n\ge 1$. For the other elements in $\mathcal{S}_{n,m}$, we have
	\begin{eqnarray}
		\begin{bmatrix}
			\alpha_{n,m}\\
			\gamma_{n,m}
		\end{bmatrix}
		=-\begin{bmatrix}
			d_{n,m}^{11} & d_{n,m}^{12} \\
			d_{n,m}^{21} & d_{n,m}^{22}
		\end{bmatrix}^{-1}
		\begin{bmatrix}
			e_{n,m}^{11} & e_{n,m}^{12} \\
			e_{n,m}^{21} & e_{n,m}^{22}
		\end{bmatrix}
		\begin{bmatrix}
			a_{n,m}\\
			c_{n,m}
		\end{bmatrix}
	\end{eqnarray}
	where 
	\begin{eqnarray}
		d_{n,m}^{11} & = &\frac{1}{i\ks R}\left(h_n^{(1)}(\ks R)+\ks R {h_n^{(1)}}'(\ks R)\right),\quad d_{n,m}^{12}  = \frac{1}{R}h_n^{(1)}(\kp R) 
		\\
		d_{n,m}^{21} & = &\frac{n(n+1)}{i\ks R}h_n^{(1)}(\ks R),\quad d_{n,m}^{22}  = \kp {h_n^{(1)}}'(\kp R). 
	\end{eqnarray}
	and $e_{n,m}^{ij}$ is simply replacing the $h_n^{(1)}$ in $d_{n,m}^{ij}$ by $j_n$.

\end{document}